\documentclass[sigconf, nonacm]{acmart}

\usepackage[ruled,vlined,linesnumbered]{algorithm2e}
\usepackage{xspace}
\usepackage{amsthm,amsmath,epsfig,graphics,enumerate,float,subfigure}
\usepackage{epstopdf}
\usepackage{svg}
\usepackage{balance}
\usepackage{color}
\usepackage{setspace}
\usepackage{bm}
\usepackage{multirow}

\newcommand{\ignore}[1]{}
\newcommand{\nop}[1]{}
\newcommand{\eat}[1]{}
\newcommand{\kw}[1]{{\ensuremath{\mathsf{#1}}}\xspace}

\newcommand{\stitle}[1]{\vspace{1ex} \noindent{\bf #1}}
\long\def\comment#1{}

\newcommand{\bd}{\kw{BD}}
\newcommand{\bindex}{\kw{BD\textrm{-}\xspace Index}}
\newcommand{\er}{\kw{ER}}
\newcommand{\eb}{\kw{EB}}
\newcommand{\pr}{\kw{PR}}
\newcommand{\gnn}{\kw{GNN}}
\newcommand{\dfs}{\kw{DFS}}

\newcommand{\snap}{\kw{SNAP}}
\newcommand{\dimacs}{\kw{DIMACS}}

\newcommand{\randomprojection}{\kw{RP}}
\newcommand{\push}{\kw{Push}}
\newcommand{\stw}{\kw{STW}}
\newcommand{\swf}{\kw{SWF}}
\newcommand{\pushp}{\kw{Push\textrm{+}}}

\newcommand{\buildtree}{\kw{BuildHierarchyTree}}
\newcommand{\rooth}{\kw{root}}

\newcommand{\anc}{\kw{Anc}}
\newcommand{\desc}{\kw{Desc}}

\newcommand{\sota}{\kw{SOTA}}

\newcommand{\fullusa}{\kw{Full\textrm{-}USA}}

\newcommand{\lapsolver}{\kw{LapSolver}}

\begin{document}
\title{BD-Index: Scalable Biharmonic Distance Queries on Large Graphs via Divide-and-Conquer Indexing}
\author{Yueyang Pan}
\affiliation{%
  \institution{Beijing Institute of Technology}
  \streetaddress{5 South Zhongguancun Street, Haidian District}
  \city{Beijing}
  \country{China}
  \postcode{100081}
}
\email{yypan@bit.edu.cn}

\author{Meihao Liao}
\affiliation{%
  \institution{Beijing Institute of Technology}
  \streetaddress{5 South Zhongguancun Street, Haidian District}
  \city{Beijing}
  \country{China}
  \postcode{100081}
}
\email{mhliao@bit.edu.cn}

\author{Rong-Hua Li}
\affiliation{%
  \institution{Beijing Institute of Technology}
  \streetaddress{5 South Zhongguancun Street, Haidian District}
  \city{Beijing}
  \country{China}
  \postcode{100081}
}
\email{lironghuabit@126.com}

\begin{abstract}
Biharmonic distance (\bd) is a powerful graph distance metric with many applications, including identifying critical links in road networks and mitigating over-squashing problem in \gnn.
However, computing \bd\ is extremely difficult, especially on large graphs.
In this paper, we focus on the problem of \emph{single-pair} \bd\ query. Existing methods mainly rely on random walk-based approaches, which work well on some graphs but become inefficient when the random walk cannot mix rapidly.
To overcome this issue, we first show that the biharmonic distance between two nodes $s,t$, denoted by $b(s,t)$, can be interpreted as the distance between two random walk distributions starting from $s$ and $t$. To estimate these distributions, the required random walk length is large when the underlying graph can be easily cut into smaller pieces. Inspired by this observation, we present novel formulas of \bd to represent $b(s,t)$ by independent random walks within two node sets $\mathcal{V}_s$, $\mathcal{V}_t$ separated by a small \emph{cut set} $\mathcal{V}_{cut}$, where $\mathcal{V}_s\cup\mathcal{V}_t\cup\mathcal{V}_{cut}=\mathcal{V}$ is the set of graph nodes. Building upon this idea, we propose \bindex, a novel index structure which follows a divide-and-conquer strategy. The graph is first cut into pieces so that each part can be processed easily. Then, all the required random walk probabilities can be deterministically computed in a bottom-top manner. When a query comes, only a small part of the index needs to be accessed. We prove that \bindex\ requires $O(n\cdot h)$ space, can be built in $O(n\cdot h\cdot (h+d_{max}))$ time, and answers each query in $O(n\cdot h)$ time, where $h$ is the height of a hierarchy partition tree and $d_{max}$ is the maximum degree, which are both usually much smaller than $n$. A striking feature of \bindex is that it is a \emph{theoretically exact} method, in contrast to existing random-walk-based approaches that only provide approximate estimates of \bd. Extensive experiments on 10 large datasets demonstrate that \bindex outperforms state-of-the-art (\sota) exact methods by at least 2 orders of magnitude in speed. It is even an order of magnitude faster than \sota approximate methods. For example, on a large road network \textsf{Road-CA} with 1,971,281  nodes and 2,766,607 edges, \bindex consumes 2 seconds while its exact (approximate) competitors takes more than 2,600 (80) seconds, with a reasonably 9.3 GB index size.  Furthermore, we also conduct two case studies to confirm the effectiveness of \bd\ in real data-mining tasks.
\end{abstract}

\maketitle

\section{Introduction}\label{sec:intro}
Let $\mathcal{G}=(\mathcal{V},\mathcal{E})$ be an undirected and connected graph with $n=|\mathcal{V}|$ nodes and $m=|\mathcal{E}|$ edges. Let $\mathbf{A}\in\mathbb{R}^{n\times n}$ be its adjacency matrix and $\mathbf{D}=\mathrm{diag}(d_1,\dots,d_n)$ be the degree matrix (a diagonal matrix) where $d_i=\sum_j \mathbf{A}_{ij}$.
The graph Laplacian is defined as $\mathbf{L}=\mathbf{D}-\mathbf{A}$.
The \emph{biharmonic distance} (\bd) between nodes $s$ and $t$ is defined as
\[
b(s,t)=(\mathbf{e}_s-\mathbf{e}_t)^{\top}\mathbf{L}^{2\dagger }(\mathbf{e}_s-\mathbf{e}_t),
\]
where $\mathbf{L}^\dagger$ denotes the Moore–Penrose pseudoinverse of $\mathbf{L}$, $\mathbf{L}^{2\dagger}$ represents the square of $\mathbf{L}^{\dagger}$, and $\mathbf{e}_i$ is the $i$-th standard basis vector~\cite{TOG2010Lipman}.
The \emph{single-pair} \bd\ query problem is: given two nodes $s,t\in V$, compute $b(s,t)$ efficiently without computing $\mathbf{L}^{\dagger}$ or all-pairs distances.

The \bd\ metric has been applied in many areas of data management and network analysis~\cite{TOG2010Lipman,IJCAI2018Yi,ACC2018Yi,TIT2022Yi,ICML2023Black,Black2024,LiuKDD2024}.
Two representative applications are (1) identifying critical link identification on road networks~\cite{IJCAI2018Yi}, and (2) over-squashing mitigation in \gnn~\cite{ICML2023Black}.
Both applications require \emph{many fast single-pair queries} rather than global computation.
For example, in \gnn\ rewiring, one needs to repeatedly evaluate \bd\ between many node pairs to add or remove edges adaptively; in road networks, detecting critical links requires repeated distance evaluation between selected intersections or regions.
Thus, the main challenge is to answer numerous single-pair \bd\ queries accurately and efficiently on large graphs.

\bd can be computed by applying exact Laplacian solvers such as Cholesky decomposition solver \cite{HornJohnson2013MatrixAnalysis} and other Laplacian solvers \cite{Cohen2014SolvingSDD, Gao2023RobustLaplacian, KyngSachdeva2016, rchol} with high precision. The most advanced Laplacian solver applies approximate Gaussian elimination to build pre-conditioners for the PCG (preconditioned conjugate gradient) routine \cite{Gao2023RobustLaplacian}, which requires $\widetilde{O}(m)$ time. However, as the hidden factor is large, it still requires high time and memory cost. To further enhance efficiency, existing approximate methods for computing \bd\ can be divided into two categories.
The first group, \emph{Laplacian solver-based methods}~\cite{TOG2010Lipman}, use random projections by iteratively solving a small number of Laplacian solvers.
The second group, \emph{random walk-based methods}~\cite{LiuKDD2024}, approximate \bd\ via sampling random walks between nodes.
These methods are lightweight and perform well on many graphs, but their efficiency strongly depends on the walk length~$l$, which must be sufficiently large to achieve small approximation error.
In practice, $l$ can be very large; for instance, on the \textsc{Amazon} dataset, $l=10^{6}$ is needed to reach relative error of $10^{-4}$, and the average query time exceeds $10^{2}$ seconds.
Therefore, while random walk approaches are effective in some cases, they can be extremely slow or inaccurate on graphs where long walks dominate.

To overcome these limitations, we present several new insights into the structure of \bd.
We show that \bd\ can be interpreted as the distance between two random walk distributions starting from nodes $s$ and $t$.
We further discover that the regions corresponding to long random walks are often easy to cut, meaning that the random walk can be restricted within two large subsets that are divided by a small cut set. This leads to a divide-and-conquer formulation: we can cut the graph into smaller pieces, compute local random walk probabilities within each piece, and then combine them hierarchically.
These observations inspire our new index structure, named \bindex, which deterministically stores intermediate random walk probabilities hierarchically.

\bindex\ follows a bottom-up construction.
We first cut the graph into subgraphs along small vertex cuts and recursively compute local random walk quantities inside each piece.
Each piece contributes partial statistics that are merged upward until the entire graph is covered.
During query time, only the relevant parts of the index along the cut paths of $s$ and $t$ need to be accessed.
Theoretical analysis shows that the index size is $O(n\cdot h)$, it can be constructed in $O(n\cdot h\cdot (h+d_{max}))$ time, each query takes $O(n\cdot h)$ time, where $h$ is the height of the hierarchy tree $\mathcal{H}$, typically much smaller than $n$ (as confirmed in our experiments). It is important to notice that although \bd is represented in terms of random walks, the proposed \bindex is an exact method, as all random walk probabilities are computed deterministically without sampling, and the only negligible error comes from floating-point precision.

We conduct extensive experiments on 10 real-world datasets.
On the \textsc{Amazon} graph, \bindex\ achieves an order-of-magnitude speedup over approximate methods with a moderate 21 GB index.
On road networks, \bindex\ performs even better—for example, on the \textsc{NewYork} graph, it requires only 0.35 GB of index space while being two orders of magnitude faster over the fastest approximate method.
Meanwhile, \bindex\ is theoretically exact. In addition, two case studies demonstrate the practical value of \bd\ in real applications—one in identifying critical links in urban transportation networks, and another in improving \gnn\ performance through \bd-guided rewiring.

\section{Preliminaries}

\begin{figure}[t]
  \centering
  \makebox[\linewidth][c]{%
    \begin{minipage}[t]{0.3\linewidth}\centering
      \includegraphics[width=\linewidth]{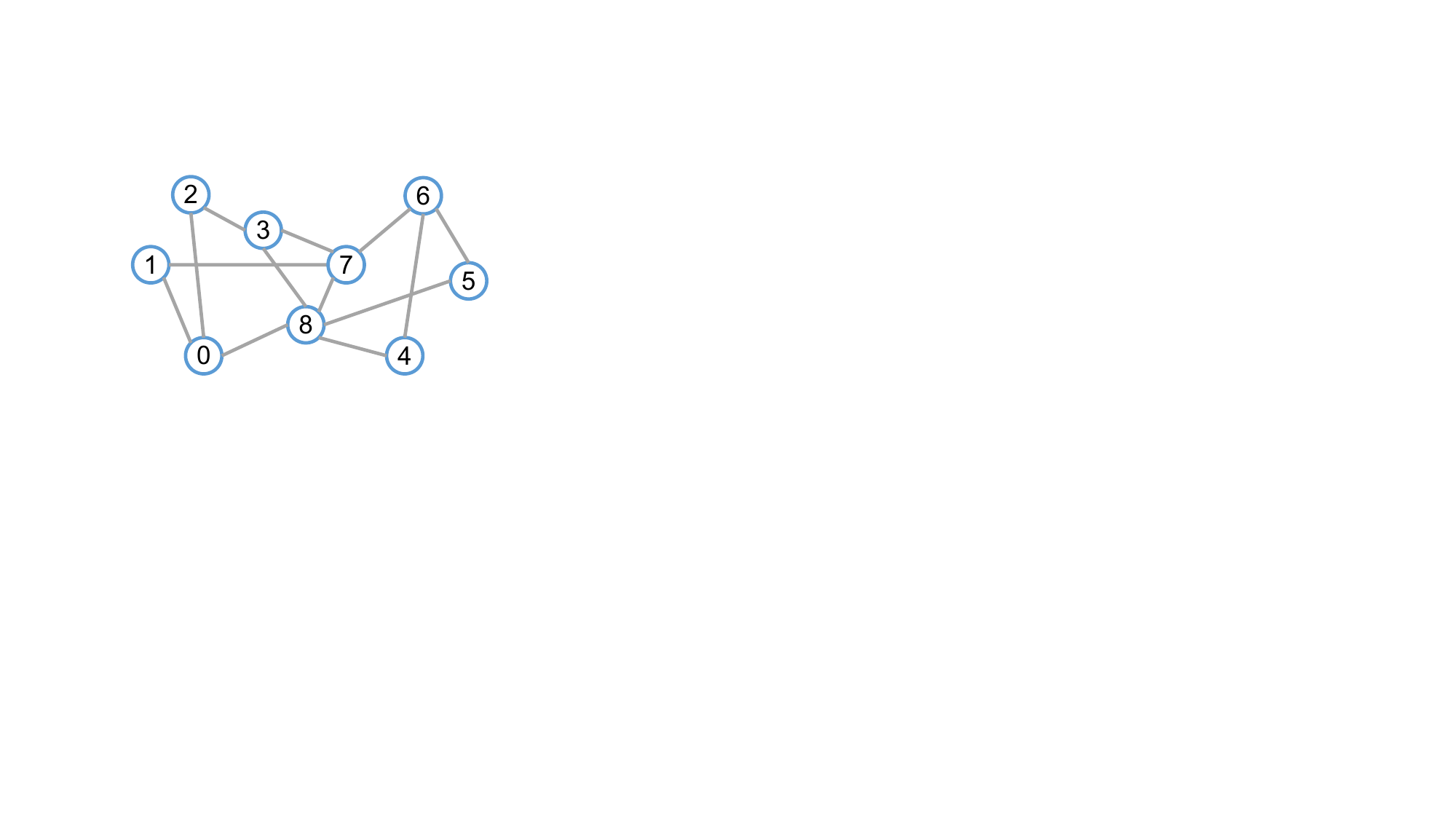}\\[-0.35em]
      \small (a) Graph $\mathcal{G}$
    \end{minipage}%
    \hspace{0.04\linewidth}
    \begin{minipage}[t]{0.66\linewidth}\centering
      \includegraphics[width=\linewidth]{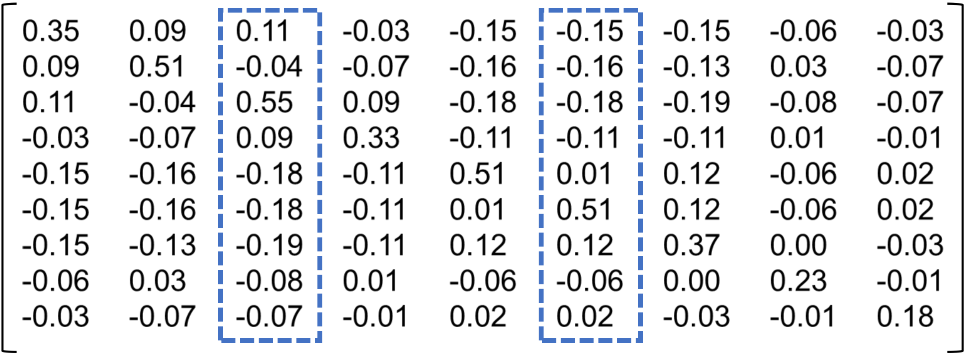}\\[-0.35em]
      \small (b) Moore--Penrose pseudoinverse of $\mathbf{L}$
    \end{minipage}%
  }

  \caption{Graph $\mathcal{G}$ and its Moore--Penrose pseudoinverse $\mathbf{L}^\dagger$. For example,
$b(2,5)
= \bigl\|L^{\dagger}\mathbf{e}_2 - L^{\dagger}\mathbf{e}_5\bigr\|_2^2
= 1.28$.
}
  \label{fig:g-and-lpinv}
\end{figure}

Let $\mathcal{G}=(\mathcal{V},\mathcal{E})$ be an undirected and connected graph with $n=|\mathcal{V}|$ nodes and $m=|\mathcal{E}|$ edges.  
The adjacency matrix is $\mathbf{A}\in\mathbb{R}^{n\times n}$, where $A_{ij}=1$ if $(i,j)\in \mathcal{E}$ and $0$ otherwise.  
The degree matrix is $\mathbf{D}=\mathrm{diag}(d_1,\dots,d_n)$ with $d_i=\sum_jA_{ij}$.  
The combinatorial Laplacian is $\mathbf{L}=\mathbf{D}-\mathbf{A}$, and $\mathbf{L}^{\dagger}$ denotes its Moore–Penrose pseudoinverse.  
The vector $\mathbf{e}_i$ represents the $i$-th standard basis vector. The biharmonic distance between nodes $s,t\in V$, denoted by $b(s,t)$, is defined by:
\[
b(s,t)=
(\mathbf{e}_s-\mathbf{e}_t)^{\!\top}\mathbf{L}^{2\dagger }(\mathbf{e}_s-\mathbf{e}_t),
\]
where $\mathbf{L}^{2\dagger}$ denotes the square of $\mathbf{L}^{\dagger}$. Computing $b(s,t)$ exactly requires access to $\mathbf{L}^{\dagger}$, which is very expensive for large graphs.

\stitle{Random walk.}  
A simple random walk on $G$ starts from a node and at each step moves to one of its neighbors chosen with equal probability.  
The transition matrix is  
$\mathbf{P}=\mathbf{D}^{-1}\mathbf{A},
P_{ij}=\frac{A_{ij}}{d_i}.$
Then $(\mathbf{P}^k)_{ij}$ is the probability that a walk beginning at node $i$ is at node $j$ after $k$ steps. The mixing time is the number of steps required for the walk’s distribution to become close to the stationary distribution. For a walk starting from node $s$, the expected number of times it visits node $t$ is  
$
\tau_{s,t}=\sum_{k=0}^{\infty}(\mathbf{P}^{k})_{st}.
$
We also define the \emph{degree-normalized expected visit count}
$\tilde{\tau}_{s,t}=\frac{\tau_{s,t}}{d_t}$,
which divides by the degree of $t$.
A \emph{$v$-absorbed random walk} is a random walk that starts from a node and stops once it first reaches node~$v$. 
For such a walk from $s$, we define the expected number of visits to node~$t$ as 
$\tau^{(v)}_{s,t}$,
and its degree-normalized counterpart as 
$\tilde{\tau}^{(v)}_{s,t} = \frac{\tau^{(v)}_{s,t}}{d_t}$.
Previous work~\cite{Landmark2023Liao} has shown that if $\mathbf{L}_v$ denotes the principal submatrix of the graph Laplacian~$\mathbf{L}$ obtained by removing 
the row and column corresponding to node~$v$, then
$\tilde{\tau}^{(v)}_{s,t} = \big(\mathbf{L}_{v}^{-1}\big)_{s,t}$.

\stitle{Single-pair query problem.}  
Given nodes $s,t\in \mathcal{V}$, a \emph{single-pair query} returns $b(s,t)$.  
In many applications such as identifying critical links in road networks~\cite{IJCAI2018Yi} or over-squashing mitigation in \gnn~\cite{ICML2023Black}, many such queries must be answered for different pairs.  
The goal is to build an index that enables efficient queries of \bd.

\subsection{Existing methods and their defects}
Exact methods for computing \bd formulate $b(s,t)$ as solving a linear system
$\mathbf{L}^{2}\mathbf{x}=\mathbf{e}_s-\mathbf{e}_t$, $b(s,t)=(\mathbf{e}_s-\mathbf{e}_t)^{\!\top}\mathbf{x}$.
A straightforward approach is to invoke a Cholesky factorization~\cite{HornJohnson2013MatrixAnalysis}, or other high-precision Laplacian solvers~\cite{Cohen2014SolvingSDD,Gao2023RobustLaplacian,KyngSachdeva2016,rchol}. The state-of-the-art solvers construct a sequence of preconditioners via approximate Gaussian elimination and use these inside a preconditioned conjugate gradient (PCG) routine~\cite{Gao2023RobustLaplacian}. This yields a nearly-linear $\widetilde{O}(m)$ running time in theory.
However, the hidden constants are large: constructing and storing multi-level preconditioners is memory-intensive and each solve still requires many PCG iterations. Consequently, these solvers remain costly on large graphs and are impractical for workloads with many single-pair BD queries, each requiring a separate linear solve.

Existing approximate methods can be grouped into two main categories: (i) \emph{Laplacian-solver–based methods.} Yi et al.~\cite{IJCAI2018Yi} use \emph{random projection}, which projects vectors into a smaller subspace to estimate $\mathbf{L}^{2\dagger }(\mathbf{e}_s-\mathbf{e}_t)$.  
Let the projection dimension be $r$,  
the total complexity becomes $O(mr+n r^{2})$ time and $O(nr)$ space.  
However, to obtain small error, $r$ must still be large, so the method is not efficient on very large graphs. (ii) \emph{Random walk–based methods}~\cite{LiuKDD2024} reformulated $b(s,t)$ in terms of random walk expectations with the transition matrix $\mathbf{P}$.  
They proposed several algorithms—\push, \stw, and the improved \swf—that compute partial walk probabilities up to a maximum length $l$.  
All of them approximate $b(s,t)$ by combining the contributions from walks of length at most $l$.  
The total time complexity of \swf is  
$O\!\left(\frac{2}{\epsilon^{2}(\min d)^{4}}\,l^{5}\right)$,
where $\epsilon$ is the target relative error and $\min d$ is the minimum node degree. When the graph has long paths or is very sparse, a large $l$ is required, which makes the query slow and less accurate.  

To overcome these problems, we first deeply investigate the random walk interpretation of \bd, showing that random walk-based approaches perform poorly when the underlying graph is easily separable. Thus, we present several new interpretations of \bd in Section~\ref{sec:new-result} to represent \bd in terms of random walks limit in two independent sets separated by a small cut set. Then, we propose a novel index-based approach \bindex in Section~\ref{sec:bd-index} to efficiently answer single-pair queries over the whole graph.

\section{New theoretical results}\label{sec:new-result}
In this section, we present several new theoretical results on \bd. We first show that previous \emph{random walk–based methods} implicitly provide a random walk interpretation of \bd: the \bd\ between $s$ and $t$ can be viewed as the distance between two random walk distributions, where more similar distributions yield smaller \bd\ values. Next, we derive a new formula, showing that \bd\  can also be explained by random walks from $s$ and $t$ to any node $v$. When $v$ separates the graph, we can restrict the walks to their respective subgraphs. Finally, we extend this to the case where the cut set $\mathcal{V}_{cut}$ contains multiple nodes: the walks can still be restricted to the two subgraphs, while the extra information can be stored with very little space. Our overall idea is illustrated in Figure~\ref{fig:bd-illustration}.

\subsection{Interpreting \bd in terms of random walk}

\begin{lemma}[Global Random Walk Representation of \bd ]
\label{lem:full-walk}
Let $\tilde{\boldsymbol{\tau}}^{(\infty)}_{s}$ denote the degree-normalized distribution of the expected visit counts for an infinite random walk starting from node $s$. Then, the biharmonic distance satisfies 
\begin{equation}\label{eq:bd-full}
b(s,t)\;=\;\|\tilde{\boldsymbol{\tau}}^{(\infty)}_{s}-\tilde{\boldsymbol{\tau}}^{(\infty)}_{t}\|_2^2 \;-\; \frac{1}{n}\,\Big(\mathbf 1^{\top}(\tilde{\boldsymbol{\tau}}^{(\infty)}_{s}-\tilde{\boldsymbol{\tau}}^{(\infty)}_{t}) \Big)^2.
\end{equation}
\end{lemma}
\begin{proof}
For any two nodes $s$ and $t$, we define $\tau_{s,t}^{(\ell)}$ as the expected
number of visits to $t$ made by a length-$\ell$ random walk starting from $s$.
Formally, let
$\mathrm{path}_\ell = (v_0, v_1, \dots, v_\ell)$
denote a walk of length $\ell$ and write
$\Pr(\mathrm{path}_\ell) = \prod_{i=0}^{\ell-1} \mathbf{P}_{v_i,v_{i+1}},$
with the convention $\Pr(\mathrm{path}_0)=1$ when $\ell=0$.
We use $[\cdot]$ for the Iverson bracket, i.e., $[{\rm cond}]=1$ if the
condition holds and $0$ otherwise.
Then
\begin{align*}
\tau_{s,t}^{(\ell)}
  &= \mathbb{E}\Bigl[\sum_{j=0}^{\ell} [v_j = t] \,\Bigm|\, v_0 = s\Bigr] 
  = \sum_{\mathrm{path}_\ell} [v_0 = s] \Pr(\mathrm{path}_\ell)
       \sum_{j=0}^{\ell} [v_j = t] \\
  &= \sum_{\mathrm{path}_\ell} [v_0 = s]
       \Bigl(\prod_{i=0}^{\ell-1} \mathbf{P}_{v_i,v_{i+1}}\Bigr)
       \sum_{j=0}^{\ell} [v_j = t] \\
  &= \sum_{j=0}^{\ell} \sum_{\mathrm{path}_\ell}
       [v_0 = s][v_j = t]
       \prod_{i=0}^{\ell-1} \mathbf{P}_{v_i,v_{i+1}} \\
  &= \sum_{j=0}^{\ell} \sum_{\mathrm{path}_\ell}
       [v_0 = s][v_j = t]
       \Bigl(\prod_{i=0}^{j-1} \mathbf{P}_{v_i,v_{i+1}}\Bigr)
       \Bigl(\prod_{i=j}^{\ell-1} \mathbf{P}_{v_i,v_{i+1}}\Bigr) \\
  &= \sum_{j=0}^{\ell}
       \sum_{\mathrm{path}_j} [v_0 = s][v_j = t] \Pr(\mathrm{path}_j)
       \sum_{\mathrm{path}_{\ell-j}} [v'_0 = t] \Pr(\mathrm{path}_{\ell-j}) \\
  &= \sum_{j=0}^{\ell} \mathbf{P}^j_{s,t}
       \sum_{\mathrm{path}_{\ell-j}} [v'_0 = t] \Pr(\mathrm{path}_{\ell-j}) \\
  &= \sum_{j=0}^{\ell} \mathbf{P}^j_{s,t}
       \sum_{k=1}^{n} \mathbf{P}^{\,\ell-j}_{t,k} 
  = \sum_{j=0}^{\ell} \mathbf{P}^j_{s,t},
\end{align*}
where $v'_0$ denotes the starting node of the suffix walk
$\mathrm{path}_{\ell-j}$ and we used that each row of
$\mathbf{P}^{\,\ell-j}$ sums to $1$.
Hence, whenever the series converges, the expected number of visits to $t$
along an infinite-length walk from $s$ is
$\tau_{s,t}^{(\infty)} := \sum_{j=0}^{\infty} \mathbf{P}^j_{s,t}.$
It is convenient to collect these expectations into a degree-normalized
visit-count vector
$\tilde{\boldsymbol{\tau}}_s
  := \sum_{i=0}^{\infty} \mathbf{e}_s \mathbf{P}^i \mathbf{D}^{-1},$
where $\mathbf{e}_s$ is the one-hot row vector for node $s$; the
$t$-th entry of $\tilde{\boldsymbol{\tau}}_s$ equals
$\tau_{s,t}^{(\infty)}/d_t$.
Previous work~\cite{LiuKDD2024} has proved that
\[
b(s,t)
  = \left\|
      \sum_{i=0}^{\infty} (\mathbf{e}_{s}-\mathbf{e}_{t})^{\top} \mathbf{P}^i \mathbf{D}^{-1}
    \right\|_2^2
    + \frac{1}{n}
      \left\|
        \sum_{i=0}^{\infty} (\mathbf{e}_{s}-\mathbf{e}_{t})^{\top} \mathbf{P}^i \mathbf{D}^{-1}\mathbf{1}
      \right\|_2^2.
\]
By substituting
$\tilde{\boldsymbol{\tau}}_s
  = \sum_{i=0}^{\infty} \mathbf{e}_s \mathbf{P}^i \mathbf{D}^{-1}$
into the above expression, we obtain the desired result, which completes the proof.
\end{proof}
Lemma~\ref{lem:full-walk} expresses \bd\ as the difference between two infinite random walk distributions.
However, in some parts of the graph, random walks mix slowly. According to the Cheeger inequality~\cite{Chung1997}, slow mixing implies a small spectral gap, which reflects the existence of a small cut. Such regions are structurally easy to separate, as random walks tend to remain within them for a long time. Therefore, we aim to cut the graph at these places.
To this end, we introduce $v$-absorbed random walk, which terminates when it hits vertex $v$. If $v$ is a cut vertex, the random walk can then be decomposed into two independent walks within the separated subgraphs.

\begin{lemma}[$v$-absorbed Random Walk Representation of \bd]
\label{lem:v-absorbed}
For any node $v\in V$, let $\tilde{\boldsymbol{\tau}}^{(v)}_{s}$ denote the degree-normalized distribution of the expected visit counts for an $v$-absorbed random walk starting from node $s$. Then, the biharmonic distance satisfies 
\begin{align}
b(s,t) 
&= \big\| \mathbf{L}_{v}^{-1}(\mathbf{e}_s - \mathbf{e}_t) \big\|_2^{2} 
   - \frac{1}{n} \left( \mathbf{1}^{\!\top} \mathbf{L}_{v}^{-1}(\mathbf{e}_s - \mathbf{e}_t) \right)^{2} \notag\\
&= \big\| \tilde{\boldsymbol{\tau}}^{(v)}_{s} - \tilde{\boldsymbol{\tau}}^{(v)}_{t} \big\|_2^{2} 
   - \frac{1}{n} \left( \mathbf{1}^{\!\top} (\tilde{\boldsymbol{\tau}}^{(v)}_{s} - \tilde{\boldsymbol{\tau}}^{(v)}_{t}) \right)^{2}.
\label{eq:bd-tau}
\end{align}
\end{lemma}

\begin{proof}
For any vector \(\mathbf{x}\in\mathbb{R}^n\), write \(\mathbf{x}^{(-v)}\in\mathbb{R}^{n-1}\) for the vector obtained by deleting its \(v\)-th entry, and let \(\mathbf{1}\) denote the all-ones vector in \(\mathbb{R}^{n-1}\).  
Set \(\mathbf{x} = \mathbf{e}_s - \mathbf{e}_t\), so that \(\mathbf{1}^\top\mathbf{x} = 0\).  
A direct calculation using the fact that \(\mathbf{L}^\dagger\) is the Moore–Penrose pseudoinverse of \(\mathbf{L}\) and that \(\mathbf{L}_v\) is nonsingular shows that, for any zero-sum \(\mathbf{x}\),
\[
(\mathbf{L}^\dagger \mathbf{x})^{(-v)}
  = \mathbf{L}_v^{-1}\mathbf{x}^{(-v)}
    - \tfrac{1}{n}\,\mathbf{1}\mathbf{1}^{\!\top}\mathbf{L}_v^{-1}\mathbf{x}^{(-v)} .
\]
Using also that for any \(u\neq v\),
\(\mathbf{e}_u^{\!\top}\mathbf{L}^\dagger \mathbf{x} = \mathbf{e}_u^{\!\top}\mathbf{L}_v^{-1}\mathbf{x}^{(-v)}\), we obtain
\begin{align*}
b(s,t)
&= \mathbf{x}^{\!\top}(\mathbf{L}^\dagger)^2\mathbf{x} \\
&= \mathbf{x}^{(-v)\top}\mathbf{L}_v^{-1}(\mathbf{L}^\dagger\mathbf{x})^{(-v)} \\
&= \mathbf{x}^{(-v)\top}\mathbf{L}_v^{-1}\Bigl(\mathbf{L}_v^{-1}\mathbf{x}^{(-v)}
     - \tfrac{1}{n}\mathbf{1}\mathbf{1}^{\!\top}\mathbf{L}_v^{-1}\mathbf{x}^{(-v)}\Bigr) \\
&= \bigl\|\mathbf{L}_v^{-1}\mathbf{x}^{(-v)}\bigr\|_2^2
   - \tfrac{1}{n}\bigl(\mathbf{1}^{\!\top}\mathbf{L}_v^{-1}\mathbf{x}^{(-v)}\bigr)^2 \\
&= \bigl\|\mathbf{L}_v^{-1}(\mathbf{e}_s-\mathbf{e}_t)\bigr\|_2^2
   - \tfrac{1}{n}\bigl(\mathbf{1}^{\!\top}\mathbf{L}_v^{-1}(\mathbf{e}_s-\mathbf{e}_t)\bigr)^2 .
\end{align*}

For the random walk representation, previous work~\cite{Landmark2023Liao} shows that the degree-normalized expected visit counts of the \(v\)-absorbed random walk satisfy \(\tilde{\tau}^{(v)}_{u,w} = (\mathbf{L}_v^{-1})_{u,w}\) for all \(u,w\neq v\).  
In particular, \(\mathbf{L}_v^{-1}\mathbf{e}_s = \tilde{\boldsymbol{\tau}}^{(v)}_s\) and \(\mathbf{L}_v^{-1}\mathbf{e}_t = \tilde{\boldsymbol{\tau}}^{(v)}_t\), hence
\(\mathbf{L}_v^{-1}(\mathbf{e}_s-\mathbf{e}_t)
   = \tilde{\boldsymbol{\tau}}^{(v)}_s - \tilde{\boldsymbol{\tau}}^{(v)}_t\).  
Substituting this into the expression above gives
\[
b(s,t)
= \bigl\|\tilde{\boldsymbol{\tau}}^{(v)}_s - \tilde{\boldsymbol{\tau}}^{(v)}_t\bigr\|_2^2
  - \tfrac{1}{n}\bigl(\mathbf{1}^{\!\top}(\tilde{\boldsymbol{\tau}}^{(v)}_s - \tilde{\boldsymbol{\tau}}^{(v)}_t)\bigr)^2,
\]
as claimed.
\end{proof}

Based on Lemma~\ref{lem:v-absorbed}, \bd\ is given by the difference between the distributions of two $v$-absorbed random walks starting from 
$s$ and $t$. If $v$ serves as a cut vertex separating $s$ and $t$, the two walks are restricted to their own subgraphs on each side of the cut.

\begin{lemma}[Cut-vertex Random Walk Representation of \bd]
\label{lem:vertex-cut}
Suppose \(v\) is a cut vertex that divides the graph into two disconnected components \(\mathcal{V}_s\) and \(\mathcal{V}_t\), with \(s \in \mathcal{V}_s\) and \(t \in \mathcal{V}_t\).  
If we order the vertices as \((\mathcal{V}_s,\,\mathcal{V}_t)\), the matrix \(\mathbf{L}_{v}^{-1}\) becomes block diagonal and can be written as 
$\mathbf{L}_{v}^{-1} =
\begin{bmatrix}
\mathbf{L}_{\mathcal{V}_s} & \mathbf{0} \\
\mathbf{0} & \mathbf{L}_{\mathcal{V}_t}
\end{bmatrix}$,
 where \(\mathbf{L}_{\mathcal{V}_s}\) is the part of \(\mathbf{L}_{v}^{-1}\) that contains the rows and columns of the vertices in \(\mathcal{V}_s \cup \{v\}\), and \(\mathbf{L}_{\mathcal{V}_t}\) is defined in the same way for \(\mathcal{V}_t \cup \{v\}\).   Let \(\tilde{\boldsymbol{\tau}}_{s}^{(v,\mathcal{V}_s)}\) and \(\tilde{\boldsymbol{\tau}}_{t}^{(v,\mathcal{V}_t)}\) denote the degree-normalized distributions of random walks starting from \(s\) and \(t\), respectively, each restricted to its corresponding subgraph \((\mathcal{V}_s \cup \{v\})\) and \((\mathcal{V}_t \cup \{v\})\) with absorption at \(v\).  
Then
\begin{align}
b(s,t) 
&= \big\| \mathbf{L}_{\mathcal{V}_s}^{-1}\mathbf{e}_s - \mathbf{L}_{\mathcal{V}_t}^{-1}\mathbf{e}_t \big\|_2^{2} 
   - \frac{1}{n} \left( \mathbf{1}^{\!\top} (\mathbf{L}_{\mathcal{V}_s}^{-1}\mathbf{e}_s - \mathbf{L}_{\mathcal{V}_t}^{-1}\mathbf{e}_t) \right)^{2} \notag\\
&= \big\| \tilde{\boldsymbol{\tau}}^{(v,\mathcal{V}_s)}_{s} - \tilde{\boldsymbol{\tau}}^{(v,\mathcal{V}_t)}_{t} \big\|_2^{2} 
   - \frac{1}{n} \left( \mathbf{1}^{\!\top} (\tilde{\boldsymbol{\tau}}^{(v,\mathcal{V}_s)}_{s} - \tilde{\boldsymbol{\tau}}^{(v,\mathcal{V}_t)}_{t}) \right)^{2}.
\label{eq:bd-vertex-cut}
\end{align}
\end{lemma}

\begin{proof}
By Lemma~\ref{lem:v-absorbed}, for any choice of absorbing node \(v\) we have
$b(s,t)
= \bigl\|\mathbf{L}_v^{-1}(\mathbf{e}_s-\mathbf{e}_t)\bigr\|_2^2
  - \tfrac{1}{n}\bigl(\mathbf{1}^{\!\top}\mathbf{L}_v^{-1}(\mathbf{e}_s-\mathbf{e}_t)\bigr)^2 .$
Since \(v\) is a cut vertex, removing \(v\) disconnects the remaining vertices into the two components \(\mathcal{V}_s\) and \(\mathcal{V}_t\).  
If we order the vertices as \((\mathcal{V}_s,\mathcal{V}_t)\), the grounded Laplacian \(\mathbf{L}_v\) is block diagonal, so its inverse has the form
$\mathbf{L}_v^{-1}
= \begin{bmatrix}
\mathbf{L}_{\mathcal{V}_s}^{-1} & \mathbf{0} \\
\mathbf{0} & \mathbf{L}_{\mathcal{V}_t}^{-1}
\end{bmatrix},$
where \(\mathbf{L}_{\mathcal{V}_s}^{-1}\) and \(\mathbf{L}_{\mathcal{V}_t}^{-1}\) are the inverses of the corresponding diagonal blocks of \(\mathbf{L}_v\).  
Identifying \(\mathbb{R}^{n-1}\) with \(\mathbb{R}^{|\mathcal{V}_s|}\oplus\mathbb{R}^{|\mathcal{V}_t|}\), and viewing \(\mathbf{L}_{\mathcal{V}_s}^{-1}\mathbf{e}_s\) and \(\mathbf{L}_{\mathcal{V}_t}^{-1}\mathbf{e}_t\) as vectors extended by zeros outside \(\mathcal{V}_s\) and \(\mathcal{V}_t\), respectively, we obtain
\begin{align*}
\mathbf{L}_v^{-1}(\mathbf{e}_s-\mathbf{e}_t)
&=
\begin{bmatrix}
\mathbf{L}_{\mathcal{V}_s}^{-1} & \mathbf{0} \\
\mathbf{0} & \mathbf{L}_{\mathcal{V}_t}^{-1}
\end{bmatrix}
\begin{bmatrix}
\mathbf{e}_s \\ -\mathbf{e}_t
\end{bmatrix}
=
\begin{bmatrix}
\mathbf{L}_{\mathcal{V}_s}^{-1}\mathbf{e}_s \\ -\mathbf{L}_{\mathcal{V}_t}^{-1}\mathbf{e}_t
\end{bmatrix}
= \mathbf{L}_{\mathcal{V}_s}^{-1}\mathbf{e}_s - \mathbf{L}_{\mathcal{V}_t}^{-1}\mathbf{e}_t .
\end{align*}
Substituting this into the expression for \(b(s,t)\) yields
\[
b(s,t)
= \bigl\|\mathbf{L}_{\mathcal{V}_s}^{-1}\mathbf{e}_s - \mathbf{L}_{\mathcal{V}_t}^{-1}\mathbf{e}_t\bigr\|_2^2
  - \tfrac{1}{n}\bigl(\mathbf{1}^{\!\top}(\mathbf{L}_{\mathcal{V}_s}^{-1}\mathbf{e}_s - \mathbf{L}_{\mathcal{V}_t}^{-1}\mathbf{e}_t)\bigr)^2,
\]
which is the first line of~\eqref{eq:bd-vertex-cut}.

For the random walk representation, previous work~\cite{Landmark2023Liao} implies \(\tilde{\boldsymbol{\tau}}^{(v)}_u = \mathbf{L}_v^{-1}\mathbf{e}_u\) for every \(u\).  
When \(u\in\mathcal{V}_s\), any \(v\)-absorbed random walk starting from \(u\) stays in \(\mathcal{V}_s\cup\{v\}\) until it is absorbed at \(v\), so the coordinates of \(\tilde{\boldsymbol{\tau}}^{(v)}_u\) on \(\mathcal{V}_t\) are zero, and its restriction to \(\mathcal{V}_s\cup\{v\}\) coincides with \(\tilde{\boldsymbol{\tau}}^{(v,\mathcal{V}_s)}_u\).  
Thus \(\tilde{\boldsymbol{\tau}}^{(v,\mathcal{V}_s)}_s\) is exactly \(\mathbf{L}_{\mathcal{V}_s}^{-1}\mathbf{e}_s\) (extended by zeros on \(\mathcal{V}_t\)), and similarly \(\tilde{\boldsymbol{\tau}}^{(v,\mathcal{V}_t)}_t\) is \(\mathbf{L}_{\mathcal{V}_t}^{-1}\mathbf{e}_t\) (extended by zeros on \(\mathcal{V}_s\)).  
Replacing \(\mathbf{L}_{\mathcal{V}_s}^{-1}\mathbf{e}_s\) and \(\mathbf{L}_{\mathcal{V}_t}^{-1}\mathbf{e}_t\) above by \(\tilde{\boldsymbol{\tau}}^{(v,\mathcal{V}_s)}_s\) and \(\tilde{\boldsymbol{\tau}}^{(v,\mathcal{V}_t)}_t\) gives exactly the identity in~\eqref{eq:bd-vertex-cut}.
\end{proof}

Lemma~\ref{lem:vertex-cut} states that if a vertex $v$ is a cut vertex that divides the graph into two disconnected components \(\mathcal{V}_s\) and \(\mathcal{V}_t\), with \(s \in \mathcal{V}_s\) and \(t \in \mathcal{V}_t\), then $b(s,t)$ can be formulated in terms of random walks restricted within \(\mathcal{V}_s\) and \(\mathcal{V}_t\).  
This division of the graph naturally prevents the slow mixing that tends to appear around small cuts.

\subsection{Cutting random walks by a cut set}\label{subsec:3.2}
On real-world graphs, a single cut vertex that separates the graph may not exist, or may be difficult to identify.
We therefore extend the previous idea to a small cut set $\mathcal{V}_{cut}=\{c_1,c_2,\dots,c_k\}$ that cut graph into disconnected components $\mathcal{V}_s$ and $\mathcal{V}_t$.
We show that $b(s,t)$ can also be represented by distribution of random walks in $\mathcal{V}_s$ and $\mathcal{V}_t$ (i.e., \(\tilde{\boldsymbol{\tau}}_{s}^{(v,\mathcal{V}_s)}\) and \(\tilde{\boldsymbol{\tau}}_{t}^{(v,\mathcal{V}_t)}\) ) additional with some matrices $\mathbf{M}$, which can also be represented by random walks.

In the following lemma, given a small cut set $\mathcal{V}_{cut}$ with a fixed order, we remove node $c_i\in \mathcal{V}_{cut}$ one by one so that, after removing all of them, the graph becomes two disconnected parts $\mathcal{V}_s$ and $\mathcal{V}_t$ containing \(s\) and \(t\).
\begin{lemma}[Cut-set Random Walk Representation of
\bd]
\label{lem:cut-set}
Let \(R^{(1)} := V\setminus \{v\}\), and after removing node \(c_j\), 
define the remaining vertex set \(R^{(j)} := R^{(j-1)} \setminus \{c_j\}\).
Reorder the Laplacian \(\mathbf{L}_{R^{(j-1)}}\) so that \(c_j\) is placed last: $\mathbf{L}_{R^{(j-1)}} =
\begin{bmatrix}
\mathbf{L}_{R^{(j)}} & -\boldsymbol a_{c_j}\\[3pt]
-\boldsymbol a_{c_j}^{\!\top} & d_{c_j}
\end{bmatrix}$, where \(\boldsymbol a_{c_j}\) is the vector of edge weights connecting \(c_j\) 
with the other vertices in \(R^{(j)}\), and \(d_{c_j}\) is the degree of \(c_j\). 
Let 
$S_{c_j} = d_{c_j} - \boldsymbol a_{c_j}^{\!\top} \mathbf{L}_{R^{(j)}}^{-1}\boldsymbol a_{c_j}$,
and define the contribution matrix
$\mathbf{M}_{c_j}^{(j)} =
\begin{bmatrix}
\mathbf{L}_{R^{(j)}}^{-1}\boldsymbol a_{c_j}\\[3pt]
1
\end{bmatrix}
S_{c_j}^{-1}
\begin{bmatrix}
\boldsymbol a_{c_j}^{\!\top}\mathbf{L}_{R^{(j)}}^{-1},\ 1
\end{bmatrix}$.  
\(\widetilde {\mathbf{M}}_{c_j}\) denotes \(\mathbf{M}_{c_j}^{(j)}\) padded with 0 
to match the full node set \(V \times V\). Then,
\begin{align}
b(s,t) 
&= 
\Big\|
\tilde{\boldsymbol{\tau}}^{(\mathcal{V}_{cut},\mathcal{V}_s)}_{s} 
    - \tilde{\boldsymbol{\tau}}^{(\mathcal{V}_{cut},\mathcal{V}_t)}_{t} 
    + \sum_{i=1}^{j}\widetilde {\mathbf{M}}_{c_i}(\mathbf{e}_s-\mathbf{e}_t)
\Big\|_2^{2} \notag\\
\quad&
- \frac{1}{n} 
\Big(
    \mathbf{1}^{\!\top} 
    \big(
        \tilde{\boldsymbol{\tau}}^{(\mathcal{V}_{cut},\mathcal{V}_s)}_{s} 
        - \tilde{\boldsymbol{\tau}}^{(\mathcal{V}_{cut},\mathcal{V}_t)}_{t} 
        + \sum_{i=1}^{j}\widetilde {\mathbf{M}}_{c_i}(\mathbf{e}_s-\mathbf{e}_t)
    \big)
\Big)^{2}.
\label{eq:bd-cut-set}
\end{align}
\end{lemma}

\begin{proof}
By construction \(\mathbf{L}_v = \mathbf{L}_{R^{(1)}}\).
For each \(i\in\{1,\dots,j\}\), reorder the vertices in \(R^{(i-1)}\) so that
\(c_i\) is placed last, and write
\[
\mathbf{L}_{R^{(i-1)}} =
\begin{bmatrix}
\mathbf{L}_{R^{(i)}} & -\boldsymbol a_{c_i}\\[2pt]
-\boldsymbol a_{c_i}^{\!\top} & d_{c_i}
\end{bmatrix},
\qquad
S_{c_i} = d_{c_i} - \boldsymbol a_{c_i}^{\!\top}\mathbf{L}_{R^{(i)}}^{-1}\boldsymbol a_{c_i}.
\]
The standard block matrix inversion formula gives
\begin{align*}
\mathbf{L}_{R^{(i-1)}}^{-1}
&=
\begin{bmatrix}
\mathbf{L}_{R^{(i)}} & -\boldsymbol a_{c_i}\\[2pt]
-\boldsymbol a_{c_i}^{\!\top} & d_{c_i}
\end{bmatrix}^{-1} \\[2pt]
&=
\begin{bmatrix}
\mathbf{L}_{R^{(i)}}^{-1} & 0\\[2pt]
0 & 0
\end{bmatrix}
+
\begin{bmatrix}
\mathbf{L}_{R^{(i)}}^{-1}\boldsymbol a_{c_i}S_{c_i}^{-1}\boldsymbol a_{c_i}^{\!\top}\mathbf{L}_{R^{(i)}}^{-1}
&
\mathbf{L}_{R^{(i)}}^{-1}\boldsymbol a_{c_i}S_{c_i}^{-1}\\[2pt]
S_{c_i}^{-1}\boldsymbol a_{c_i}^{\!\top}\mathbf{L}_{R^{(i)}}^{-1}
&
S_{c_i}^{-1}
\end{bmatrix} \\[2pt]
&=
\begin{bmatrix}
\mathbf{L}_{R^{(i)}}^{-1} & 0\\[2pt]
0 & 0
\end{bmatrix}
+ \mathbf{M}_{c_i}^{(i)} .
\end{align*}
Padding by zeros to the full vertex set \(V\times V\), this can be written as
\(\mathbf{L}_{R^{(i-1)}}^{-1} = \widetilde{\mathbf{L}}_{R^{(i)}}^{-1} + \widetilde{\mathbf{M}}_{c_i}\).
Iterating from \(i=1\) to \(i=j\) yields
\[
\mathbf{L}_v^{-1}
= \mathbf{L}_{R^{(1)}}^{-1}
= \widetilde{\mathbf{L}}_{R^{(j)}}^{-1} + \sum_{i=1}^{j}\widetilde{\mathbf{M}}_{c_i}.
\]

After removing the cut set \(\mathcal{V}_{cut}\), the remaining vertices split into
two components \(\mathcal{V}_s\) and \(\mathcal{V}_t\) containing \(s\) and \(t\), respectively.
By Lemma~\ref{lem:vertex-cut}, the matrix \(\mathbf{L}_{R^{(j)}}^{-1}\) is block diagonal:
$\mathbf{L}_{R^{(j)}}^{-1}
=
\begin{bmatrix}
\mathbf{L}_{\mathcal{V}_s}^{-1} & 0\\
0 & \mathbf{L}_{\mathcal{V}_t}^{-1}
\end{bmatrix},$
so that, viewing the blocks as zero-padded to \(V\),
\begin{align*}
\mathbf{L}_v^{-1}(\mathbf{e}_s-\mathbf{e}_t)
&= \widetilde{\mathbf{L}}_{R^{(j)}}^{-1}(\mathbf{e}_s-\mathbf{e}_t)
   + \sum_{i=1}^{j}\widetilde{\mathbf{M}}_{c_i}(\mathbf{e}_s-\mathbf{e}_t) \\
&=
\begin{bmatrix}
\widetilde{\mathbf{L}}_{\mathcal{V}_s}^{-1} & 0\\
0 & \widetilde{\mathbf{L}}_{\mathcal{V}_t}^{-1}
\end{bmatrix}
(\mathbf{e}_s-\mathbf{e}_t)
+ \sum_{i=1}^{j}\widetilde{\mathbf{M}}_{c_i}(\mathbf{e}_s-\mathbf{e}_t) \\
&= \widetilde{\mathbf{L}}_{\mathcal{V}_s}^{-1}\mathbf{e}_s
   - \widetilde{\mathbf{L}}_{\mathcal{V}_t}^{-1}\mathbf{e}_t
   + \sum_{i=1}^{j}\widetilde{\mathbf{M}}_{c_i}(\mathbf{e}_s-\mathbf{e}_t).
\end{align*}
By Lemma~\ref{lem:v-absorbed},
\(\widetilde{\mathbf{L}}_{\mathcal{V}_s}^{-1}\mathbf{e}_s\) and
\(\widetilde{\mathbf{L}}_{\mathcal{V}_t}^{-1}\mathbf{e}_t\) are exactly the
degree-normalized expected visit distributions
\(\tilde{\boldsymbol{\tau}}^{(\mathcal{V}_{cut},\mathcal{V}_s)}_{s}\) and
\(\tilde{\boldsymbol{\tau}}^{(\mathcal{V}_{cut},\mathcal{V}_t)}_{t}\), respectively.
Therefore
\[
\mathbf{L}_v^{-1}(\mathbf{e}_s-\mathbf{e}_t)
=
\tilde{\boldsymbol{\tau}}^{(\mathcal{V}_{cut},\mathcal{V}_s)}_{s}
-
\tilde{\boldsymbol{\tau}}^{(\mathcal{V}_{cut},\mathcal{V}_t)}_{t}
+
\sum_{i=1}^{j}\widetilde{\mathbf{M}}_{c_i}(\mathbf{e}_s-\mathbf{e}_t).
\]

Finally, substituting this expression into the \(v\)-absorbed representation
of \bd in Lemma~\ref{lem:v-absorbed},
gives exactly the identity in~\eqref{eq:bd-cut-set}.
\end{proof}

In Lemma~\ref{lem:cut-set}, $\tilde{\boldsymbol{\tau}}^{(\mathcal{V}_{cut},\mathcal{V}_s)}_{s}$ and $\tilde{\boldsymbol{\tau}}^{(\mathcal{V}_{cut},\mathcal{V}_s)}_{s}$ represents two independent 
$v$-absorbed random walks, one inside each subgraph $\mathcal{V}_s,\mathcal{V}_t$.
Each contribution matrix $\widetilde {\mathbf{M}}_{c_j}$ deterministically adds the paths that go through node $c_j \in \mathcal{V}_{cut}$ while avoiding all the previous random paths.
Together, these terms recover the full $v$-absorbed random walk behavior on the original graph. This insight motivates our approach: we iteratively cut the graph into smaller pieces and deterministically store the corresponding random walk distributions.

\begin{figure}[t]
  \centering

  \makebox[\linewidth][c]{%
    \begin{minipage}[t]{0.35\linewidth}\centering
      \includegraphics[width=\linewidth]{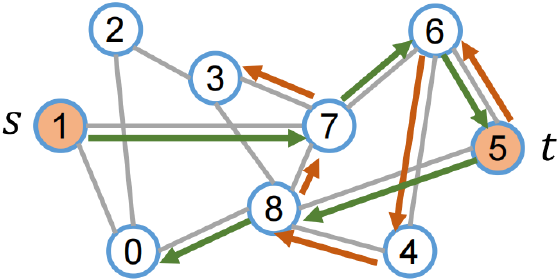}\\[-0.3em]
      (a)
    \end{minipage}%
    \hspace{0.06\linewidth}
    \begin{minipage}[t]{0.35\linewidth}\centering
      \includegraphics[width=\linewidth]{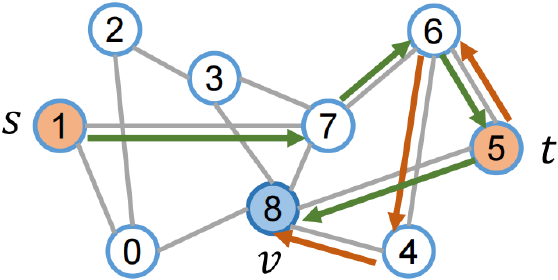}\\[-0.3em]
      (b)
    \end{minipage}%
  }

  \makebox[\linewidth][c]{%
    \begin{minipage}[t]{0.36\linewidth}\centering
      \includegraphics[width=\linewidth]{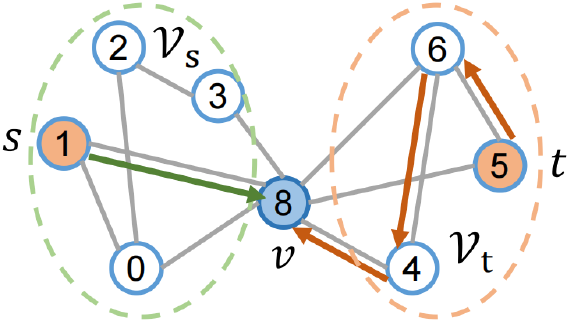}\\[-0.3em]
      (c)
    \end{minipage}%
    \hspace{0.06\linewidth}%
    \begin{minipage}[t]{0.36\linewidth}\centering
      \includegraphics[width=\linewidth]{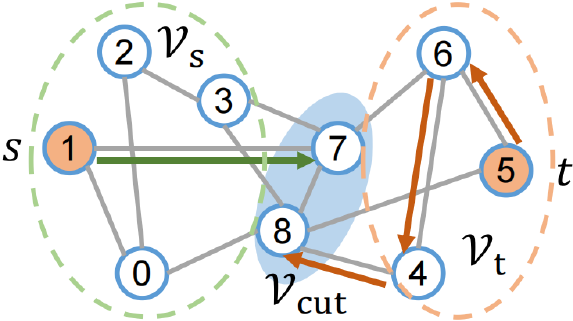}\\[-0.3em]
      (d)
    \end{minipage}%
  }
  \caption{Illustration of the proposed formulas of \bd. (a) \bd can be interpreted by random walks from $s$ and $t$ on the whole graph; (b) \bd can be interpreted by random walks from $s$ and $t$ until hitting $v$; (c) If $v$ is a cut vertex, \bd can be interpreted by random walks independently on $\mathcal{V}_s$ and $\mathcal{V}_t$; (d) \bd can be interpreted by random walks independently on $\mathcal{V}_s$ and $\mathcal{V}_t$, separated by a small cut set $\mathcal{V}_{cut}=\{v_7,v_8\}$.}
  \label{fig:bd-illustration}
\end{figure}

\section{The proposed approach: BD-Index}\label{sec:bd-index}
\begin{figure*}
    \centering
    \includegraphics[width=0.99\linewidth]{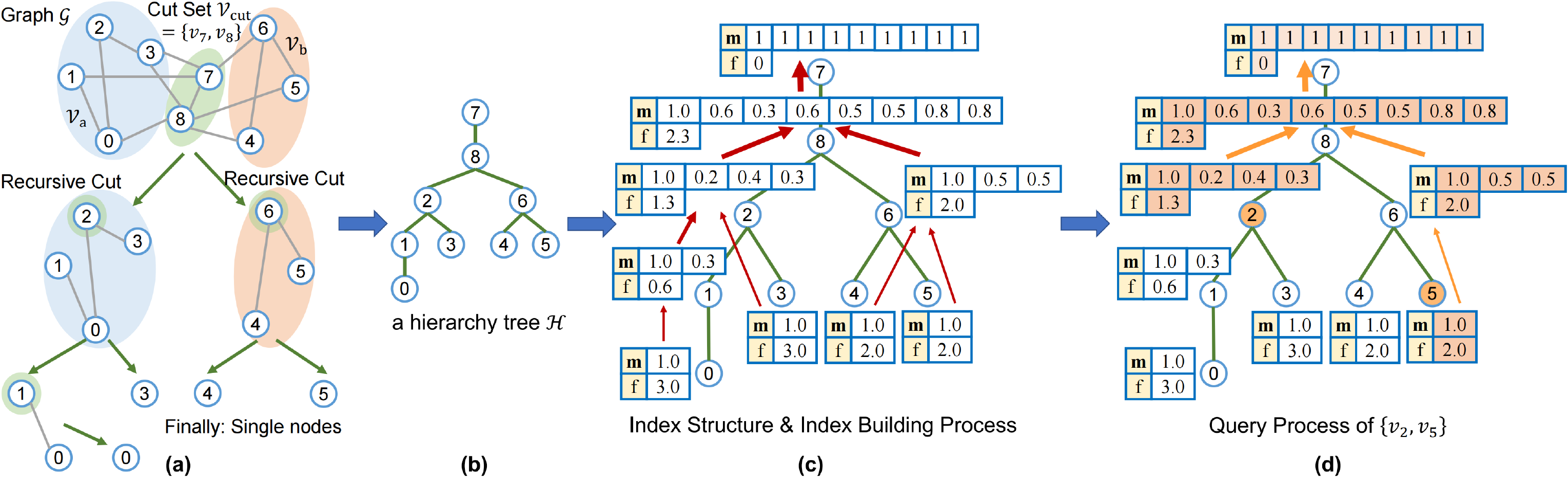}
    \caption{Illustration of \bindex's index structure, index building process and query process}
    \label{fig:procedure4all}
\end{figure*}

\subsection{High-level idea}
\stitle{Divide: cut graph into pieces.}
Building upon Lemma~\ref{lem:cut-set}, the \bd\ can be expressed as two parts:  
(i) the difference between the degree-normalized random walk distributions inside the two subgraphs \(\mathcal{V}_s\) and \(\mathcal{V}_t\), and  
(ii) a term that contains \(k = |\mathcal{V}_{\mathrm{cut}}|\) contribution matrices
\(\sum_{j=1}^{k}\widetilde{\mathbf{M}}_{c_j}\),  
which together deterministically describe the distribution of random walks that pass through the cut set \(\mathcal{V}_{\mathrm{cut}}\).  
By cutting the graph with a small cut set \(\mathcal{V}_{\mathrm{cut}}\),  
we can transform random walks on the full graph into random walks restricted in the two subgraphs  
\(\mathcal{V}_s\) and \(\mathcal{V}_t\), together with a small number of matrices \(\widetilde{\mathbf{M}}_{c_j}\).  
When the cut set size \(k\) is small, these matrices can be computed efficiently.

However, a single cut divides the graph into only two parts.  
It is easy to observe that the same property holds recursively for each subgraph:  
the random walk distribution within any subgraph can again be decomposed  
by using another small cut set that splits the subgraph into smaller pieces.  
Repeating this process eventually leads to subgraphs that are so small  
that their degree-normalized random walk distributions can be computed exactly.  
In the case where a subgraph contains only one node \(v\),  
the random walk distribution equals \(1\),  
and its degree-normalized form is \(1/d_v\).  
Therefore, by recursively partitioning the graph with a series of small cut sets until each block has only one vertex,  
we can deterministically compute the random walk distributions on the entire graph  
with very low computational cost.

Then, a natural question arises: how to store all the contribution matrices $\widetilde{\mathbf{M}}$ efficiently?  At first sight, storing all contribution matrices \(\{\widetilde{\mathbf{M}}_{c_j}\}\) seems to require \(O(n^3)\) space.  
However, each matrix \(\widetilde{\mathbf{M}}_{c_j}\), which represents all random walks that pass through the cut set,  
can be decomposed into three parts:  
(i) walks that start from nodes in \(\mathcal{V}_s \cup \mathcal{V}_t\) and first reach the cut set \(\mathcal{V}_{\mathrm{cut}}\);  
(ii) walks that start from the cut set and reach nodes in \(\mathcal{V}_s \cup \mathcal{V}_t\); and 
(iii) walks that start and end within the cut set itself.  
The first two parts can be represented by one vector of length at most \(n\),  
while the third part corresponds to a constant value shared by all node pairs \(s,t \in \mathcal{V}_s \cup \mathcal{V}_t\).  
Hence, the information in each matrix \(\widetilde{\mathbf{M}}_{c_j}\) can be stored using only one vector and one constant,  
which greatly reduces the total space needed to represent all random walk distributions.

\begin{lemma}[Compact Representation of the Contribution Matrices]\label{lem:calM}
Consider one cut node \(x \in \mathcal{V}_{\mathrm{cut}}\) in the process described in Lemma~\ref{lem:cut-set}.  
Let \(R\) be the current set of nodes that still contains \(x\),  
and let \(L_R\) denote the submatrix of the Laplacian on the nodes in \(R \setminus \{x\}\).  
Let \(\boldsymbol{a}_x\) be the vector that contains the edges between \(x\) and its neighboring nodes in \(R \setminus \{x\}\),  
and let \(d_x\) be the degree of node \(x\).

Then, the contribution matrix of \(x\) can be written as the product of two vectors:
\begin{equation}\label{eq:calM}
\mathbf{M}_x
=
\begin{bmatrix}
\boldsymbol{m}_x \\[2pt] 1
\end{bmatrix}
\, f_x^{-1}\,
\begin{bmatrix}
\boldsymbol{m}_x^{\!\top} & 1
\end{bmatrix},
\end{equation}
where
\begin{equation*}\label{eq:mx-fx}
\boldsymbol{m}_x = L_R^{-1}\boldsymbol{a}_x,
\qquad
f_x = d_x - \boldsymbol{a}_x^{\!\top} L_R^{-1}\boldsymbol{a}_x.
\end{equation*}
\end{lemma}

\begin{proof}
Consider the Laplacian on the current node set \(R\), reordered so that \(x\) is placed last.
By definition of \(L_R\), \(\boldsymbol{a}_x\), and \(d_x\), this Laplacian has the block form
$\widehat{\mathbf{L}}_R
=
\begin{bmatrix}
L_R & -\boldsymbol{a}_x \\[2pt]
-\boldsymbol{a}_x^{\!\top} & d_x
\end{bmatrix}.$
Since \(L_R\) is nonsingular, we can invert \(\widehat{\mathbf{L}}_R\) using the standard
block-matrix inversion formula.  With
\(A=L_R\), \(B=-\boldsymbol{a}_x\), \(C=-\boldsymbol{a}_x^{\!\top}\), and \(D=d_x\),
the Schur complement of \(A\) is
\[
S = D - C A^{-1} B
  = d_x - \boldsymbol{a}_x^{\!\top} L_R^{-1}\boldsymbol{a}_x
  = f_x.
\]
The inverse is then
\begin{align*}
\widehat{\mathbf{L}}_R^{-1}
&=
\begin{bmatrix}
L_R^{-1} + L_R^{-1} B S^{-1} C L_R^{-1} & -L_R^{-1} B S^{-1} \\[2pt]
- S^{-1} C L_R^{-1} & S^{-1}
\end{bmatrix} \\
&=
\begin{bmatrix}
L_R^{-1} & 0 \\[2pt]
0 & 0
\end{bmatrix}
+
\begin{bmatrix}
L_R^{-1}\boldsymbol{a}_x \\[2pt]
1
\end{bmatrix}
f_x^{-1}
\begin{bmatrix}
\boldsymbol{a}_x^{\!\top} L_R^{-1},\ 1
\end{bmatrix}.
\end{align*}
By definition (cf.\ Lemma~\ref{lem:cut-set}), the second term is exactly the
contribution matrix \(\mathbf{M}_x\) associated with eliminating node \(x\).  
Writing \(\boldsymbol{m}_x = L_R^{-1}\boldsymbol{a}_x\) gives
\[
\mathbf{M}_x
=
\begin{bmatrix}
\boldsymbol{m}_x \\[2pt] 1
\end{bmatrix}
f_x^{-1}
\begin{bmatrix}
\boldsymbol{m}_x^{\!\top} & 1
\end{bmatrix},
\]
which is \eqref{eq:calM}.
\end{proof}

By Lemma~\ref{lem:calM}, all information of every step in Lemma~\ref{lem:cut-set} can be kept by storing only the pair \((\boldsymbol{m}_x, f_x)\),  
where \(\boldsymbol{m}_x\) is kept on the nodes of \(R\) that are connected to \(x\).  
This form avoids building the whole matrix \(\mathbf{M}_x\)  
and keeps exactly the same contribution.

We first build a recursive cut hierarchy, denoted by $\mathcal{H}$, by cutting the graph until all leaves become single vertices.  
Any effective routine can be applied; we adopt a recursive minimum vertex cut~\cite{Mincut1998}, which repeatedly finds a small near-minimum cut at each level, as the default, and also evaluate a degree-based heuristic cut~\cite{TreeDecomposion1991} that splits the graph around nodes with small degrees.

\stitle{Conquer: a bottom-up method.}
After dividing the graph, we obtain a hierarchical structure $\mathcal{H}$.  
For each node $v$, its contribution matrix $\widetilde{\mathbf{M}}_v$ depends only on its descendants, denoted $\desc (v)$, that is, all nodes below it in the hierarchy.  
Hence, the contribution matrices can be computed in a bottom-up manner, where each node $v$ is updated from the information of its descendants.
\begin{lemma}[Bottom-Up Aggregation]
Let $(\mathbf{m}_u, f_u)$ be the stored pair for each $u \in \desc (v)$.  
Then the pair for node $v$ can be obtained by combining the pairs of all its descendants:
\begin{equation}\label{eq:calmf}
\begin{aligned}
\mathbf m_v
&= \mathbf e_v
   + \sum_{u\in \mathrm{Desc}(v)}
     \left(
       \frac{\sum_{x\in N(v)\cap \mathrm{Desc}(u)} \mathbf m_u[x]}{\,f_u\,}
     \right)\mathbf m_u,\\[4pt]
f_v
&= d_v - \sum_{x\in N(v) \cap \mathrm{Desc}(v)} \mathbf m_v[x].
\end{aligned}
\end{equation}

\end{lemma}
\begin{proof}
For every descendant \(u\in\mathrm{Desc}(v)\), the stored pair
\((\mathbf m_u,f_u)\) summarizes the total contribution of all cut nodes in
the subtree rooted at \(u\): the sum of their contribution matrices has the
rank–one form
\[
\mathbf M_u
=
\begin{bmatrix}
\mathbf m_u\\[2pt] 1
\end{bmatrix}
f_u^{-1}
\begin{bmatrix}
\mathbf m_u^{\!\top} & 1
\end{bmatrix},
\]
supported only on \(\mathrm{Desc}(u)\cup\{v\}\).
Since different subtrees are disjoint, the reduced Laplacian on
\(\{v\}\cup\mathrm{Desc}(v)\) after eliminating all descendants can be
obtained by adding, for each \(u\in\mathrm{Desc}(v)\), the rank–one update
\(\mathbf M_u\) to the original block Laplacian.

Now consider eliminating \(v\) in this reduced graph.  Let
\(\widehat{R} = \mathrm{Desc}(v)\) and let \(\widehat{L}_{\widehat{R}}\) be
the Laplacian on \(\widehat{R}\) after all these updates.
By construction, the new adjacency between \(v\) and the nodes in
the subtree of \(u\) is the sum of the entries in the last column of
\(\mathbf M_u\) on the neighbors of \(v\), i.e.,
$\sum_{x\in N(v)\cap \mathrm{Desc}(u)} \frac{\mathbf m_u[x]}{f_u} .$
Therefore the updated adjacency vector from \(v\) to
\(\widehat{R}=\mathrm{Desc}(v)\) can be written as
\[
\boldsymbol a_v
= \sum_{u\in\mathrm{Desc}(v)}
   \left(
     \frac{\sum_{x\in N(v)\cap \mathrm{Desc}(u)} \mathbf m_u[x]}{f_u}
   \right)\mathbf m_u .
\]

Applying Lemma~\ref{lem:calM} once more to node \(v\), with
\(\widehat{L}_{\widehat{R}}\) and \(\boldsymbol a_v\), we obtain
\[
\mathbf m_v
= \widehat{L}_{\widehat{R}}^{-1}\boldsymbol a_v
= \sum_{u\in\mathrm{Desc}(v)}
   \left(
     \frac{\sum_{x\in N(v)\cap \mathrm{Desc}(u)} \mathbf m_u[x]}{f_u}
   \right)\mathbf m_u
   + \mathbf e_v,
\]
where the term \(\mathbf e_v\) accounts for the unit entry at \(v\) itself.
This is exactly the first line in \eqref{eq:calmf}.  The corresponding
Schur complement at \(v\) is
\[
f_v
= d_v - \boldsymbol a_v^{\!\top}\mathbf m_v
= d_v - \sum_{x\in N(v)\cap\mathrm{Desc}(v)} \mathbf m_v[x],
\]
which is the second line in \eqref{eq:calmf}.  Hence the pair
\((\mathbf m_v,f_v)\) is obtained by aggregating the pairs of all
descendants, as claimed.
\end{proof}

\subsection{Index structure}
A rooted tree is a hierarchical structure with a single distinguished node called the \emph{root}.  
Each edge connects a node to one of its lower nodes, forming a parent–child relation.  
A \emph{branch} refers to a point in the tree where a node has more than one child.  
For any node $v$, we use $\anc(v)$ to denote the set of all its ancestors, that is, the nodes on the path from the root to $v$ (including $v$ itself).  
Similarly, we use $\desc(v)$ to denote the set of all its descendants, that is, the nodes in the subtree rooted at $v$ (including $v$ itself).  

The index structure of \bindex is to store the pair $(\mathbf{m}_v, f_v)$ for each node $v$. The cut hierarchy $\mathcal{H}$ is organized as a rooted tree, where each branching corresponds to a cut. The vertices within each cut set are arranged in a fixed order, consistent with the order defined in Section~\ref{subsec:3.2}.

According to Lemma~\ref{lem:calM}, each node $v$ is associated with a vector $\mathbf{m}_v$ and a scalar $f_v$.  
The length of $\mathbf{m}_v$ equals the size of $\desc (v)$, that is, the number of nodes in the subtree rooted at $v$.  
To avoid storing a full $n$-dimensional vector with many empty entries, we construct $\mathbf{m}_v$ as a compact vector of size $|\desc (v)|$.  
We implicitly index its entries using the depth-first search (\dfs) order of the hierarchy tree, so that the $i$-th entry of $\mathbf{m}_v$ corresponds to the $i$-th node in the \dfs\ sequence under $v$.  
For example, in Figure~\ref{fig:procedure4all}(c), \bindex on node $v_2$ stores: its parent node number $v_8$, a constant $f_{v_2}=1.3$, and a vector $\mathbf{m}_{v_2}$. Specifically, $\mathbf{m}_{v_2}[0]=1.0$ for $v_2$ itself,  
$\mathbf{m}_{v_2}[1]=0.2$ for $v_1$,  
$\mathbf{m}_{v_2}[2]=0.4$ for $v_0$,  
and $\mathbf{m}_{v_2}[3]=0.3$ for $v_3$,  
where the entries follow the \dfs order under $v_2$.
\begin{lemma}[Space Complexity of the Index Structure]
Let $h$ be the height of the hierarchy tree $\mathcal{H}$.  
The total space required to store all pairs $(\mathbf{m}_v, f_v)$ in the index structure is $O(n\cdot h)$.
\end{lemma}
\begin{proof}
By construction, $|\mathbf m_v|=|\mathrm{Desc}(v)|$. Summing over all $v$ counts each vertex $x$ once for each ancestor of $x$:
\[
\sum_{v} |\mathrm{Desc}(v)|
=\sum_{x} |\mathrm{Anc}(x)|
\le n h.
\]
Each pair $(\mathbf m_v,f_v)$ thus costs $O(|\mathrm{Desc}(v)|+1)$ space; the total is $O(n\cdot h)$.
\end{proof}
In real-world graphs, $h$ is typically much smaller than $n$ (as confirmed in our experiments), resulting in a small size of our index structure. For example, on \fullusa dataset ($n=23,947,348$, $m=28,854,319$), our index size  is only 168GB. 
\subsection{Index building}
\label{sec:subsecib}

The index construction begins with the hierarchy tree $\mathcal{H}$.  
Starting from the original graph, we recursively find small cut sets to divide the graph into subgraphs. For each cut set, the vertices are arranged into a short chain following the same order as defined in Section~\ref{subsec:3.2}.  
Each subgraph separated by the cut set is then attached as a child subtree.  
Repeating this process recursively yields the complete hierarchy tree $\mathcal{H}$ (corresponding to Line~14 in Algorithm~\ref{alg:indexbuilding}). As finding an optimal cut hierarchy is challenging, we adpot two heuristics: degree-based heuristic and recursive minimum cut heuristic. 

\stitle{Recursive minimum cut heuristic.} The detail of the minimum cut-based method is illustrated in Algorithm~\ref{alg:minimum-cut}. Given a graph $\mathcal{G}$, the algorithm constructs a hierarchy tree $\mathcal{H}$ in a recursive manner. 
For the input graph $\mathcal{G}$, we first call the function \textsc{GetApproxCutSet} to obtain an approximate vertex cut set. 
The vertices in this cut set are connected sequentially to form a chain in the hierarchy tree $\mathcal{H}$. 
Removing the cut set divides $\mathcal{G}$ into $k$ subgraphs (Lines~9-10).   
For each subgraph, we again apply \textsc{GetApproxCutSet} to compute its own vertex cut set, which is linked into another chain and attached as a child subtree below the corresponding node of the previous chain in $\mathcal{H}$ (Lines~11-12). 
This recursive process continues until a subgraph contains only a single vertex, at which point \textsc{GetApproxCutSet} is no longer invoked (Lines~6-8, 13-14). 

Finding the minimum vertex cut of a graph is NP-hard~\cite{np}. 
In our implementation (Line~9, function \textsc{GetApproxCutSet}), we adopt the widely used approximation algorithm \textsc{MEITS}~\cite{METIS,Mincut1998}, 
whose computational complexity is $O(m)$.
Since our hierarchy construction algorithm applies this process recursively at each level, the total running time is proportional to the number of edges times the hierarchy depth, giving an overall complexity of $O(m\cdot h)$.

\begin{algorithm}[t!]
\small
\caption{Minimum cut-based hierarchy construction}
\label{alg:minimum-cut}
\KwIn{Graph $\mathcal{G}=(\mathcal{V},\mathcal{E})$}
\KwOut{Hierarchy tree $\mathcal{H}$}

\SetKwFunction{FMain}{BuildHierarchyTree}
\SetKwFunction{FRec}{BuildSubtree}
\SetKwProg{Fn}{Function}{:}{}

\Fn{\FMain{$\mathcal{G}$}}{
    Initialize an empty hierarchy tree $\mathcal{H}$\;
    $\FRec(\mathcal{G}, null)$\;
    \KwRet{$\mathcal{H}$}\;
}

\Fn{\FRec{$\mathcal{G}=(\mathcal{V},\mathcal{E})$, parent node $p$}}{
    \If{$|\mathcal{V}|=1$}{
        Add a node $v$ labeled by the single vertex in $\mathcal{V}$ under parent $p$ in $\mathcal{H}$\;
        \KwRet{}\;
    }

    $C \gets \textsc{GetApproxCutSet}(\mathcal{G})$ \tcp*{Approximate minimum vertex cut returned by METIS~\cite{METIS}}
    Add all vertices in $C$ as a chain under parent $p$ in $\mathcal{H}$\;
    Let $b$ be the bottom node of this chain (if $C=\emptyset$, $b \gets p$)\;

    $\{\mathcal{G}_1, \mathcal{G}_2, \ldots, \mathcal{G}_k\} \gets$ connected components of the induced subgraph $\mathcal{G}[\mathcal{V} \setminus C]$\;
    \ForEach{$\mathcal{G}_i$}{
        $\FRec(\mathcal{G}_i, b)$\;
    }
}
\end{algorithm}

\begin{algorithm}[t]
\small
\caption{Minimum degree-based hierarchy construction \cite{TreeDecomposion1991}}
\label{alg:minimum-deg}
\KwIn{Graph $\mathcal{G}=(\mathcal{V},\mathcal{E})$}
\KwOut{Hierarchy tree $\mathcal{H}$}

Initialize an empty hierarchy tree $\mathcal{H}$\;
$\mathcal{G}'\gets \mathcal{G}$;

\While{$\mathcal{V}(\mathcal{G}')\neq \emptyset$}{
select $v\in \mathcal{V}(\mathcal{G}')$ with minimum degree in $\mathcal{G}'$\;
$\mathcal{N}\gets {u\in \mathcal{V}(\mathcal{G}') \mid (u,v)\in \mathcal{E}(\mathcal{G}')}$\;
create a new node $h$ in $\mathcal{H}$ representing ${v}\cup\mathcal{N}$\;
\If{$\mathcal{H}$ is not empty}{
connect $h$ to the node $h^\star\in\mathcal{H}$ that shares the largest overlap with ${v}\cup\mathcal{N}$\;
}
add fill-in edges to make $\mathcal{N}$ a clique in $\mathcal{G}'$\;
remove $v$ and its incident edges from $\mathcal{G}'$\;
}
\Return $\mathcal{H}$;

\end{algorithm}

\stitle{Minimum degree heuristic.} The detail of the minimum degree-based method~\cite{TreeDecomposion1991} is illustrated in Algorithm~\ref{alg:minimum-deg}.
Starting from the input graph $\mathcal{G}$, the algorithm repeatedly removes the vertex with the smallest degree from the working graph $\mathcal{G}'$.
At each step, it forms the set ${v}\cup\mathcal{N}$, where $\mathcal{N}$ is the set of neighbors of $v$, and creates a node $h$ in the hierarchy tree $\mathcal{H}$ to represent this set.
Then $h$ is connected to the existing node in $\mathcal{H}$ that shares the most common vertices with it.
Before removing $v$, the algorithm adds fill-in edges so that $\mathcal{N}$ becomes a clique in $\mathcal{G}'$, making ${v}\cup\mathcal{N}$ a small group of nodes that divides the remaining graph into several parts.
Because each new node is linked to exactly one parent, $\mathcal{H}$ forms a tree. The worst-case time complexity of the minimum degree–based hierarchy construction is $O(n\cdot m)$, but in our experiments its practical running time remains efficient and acceptable for all graphs tested.

From a structural view, the top-down chain in $\mathcal{H}$ records the cut sets created step by step, while each branch growing from a node on this chain corresponds to a subgraph formed when that cut divides the graph.
This process continues inside every subgraph until no vertices are left.
As a result, the minimum degree-based construction produces a hierarchy tree whose vertical chains show the sequence of cut sets, and whose branches represent the subgraphs obtained after each cut, matching the organization of the algorithm.

For example, in Figure~\ref{fig:procedure4all}(a), we first find a minimum vertex cut $\{v_7, v_8\}$. After removing $v_7$ and $v_8$, the graph splits into two disconnected parts. We fix an order on the cut set, say $(v_7, v_8)$, then root of the hierarchy tree $\mathcal{H}$ is $v_7$, and $v_7$ has a single child $v_8$.
Next, we search for a minimum vertex cut in the two parts remaining. The cut sets are $\{v_2\}$ and $\{v_6\}$. We add $v_2$ and $v_6$ as the children of $v_8$ in $\mathcal{H}$.
We repeat this process on every new part and stop when every remaining part is a single vertex. The resulting tree shown in Figure~\ref{fig:procedure4all}(b) is the hierarchy $\mathcal{H}$.  

\begin{algorithm}[t!]
\small
\caption{Index building algorithm of \bindex} 
\label{alg:indexbuilding}
\KwIn{Graph $\mathcal{G}=(\mathcal{V},\mathcal{E})$}
\KwOut{Hierarchy tree $\mathcal{H}$, $(\mathbf{m}_v, f_v)$ for each node $v\in\mathcal{V}$}

\SetKwFunction{FSub}{BuildNode}
\SetKwProg{Fn}{Function}{:}{}
\Fn{\FSub{$v$}}{
    \lIf{$v$ has been built}{\KwRet{$(\mathbf{m}_v, f_v)$}}
    $\mathbf{m}_v \gets \emptyset, f_v \gets 0$\;
    \ForEach{$u \in \desc (v)$}{  
        $(\mathbf{m}_u,f_u) \gets  \FSub (u)$\;  
    }

    \ForEach{$u \in \desc (v)$}{  
        $\mathcal{C} \gets N(v) \cap \desc (u)$\;  
        $\sigma \gets \sum_{x \in \mathcal{C}} \mathbf{m}_u[x]$\; 
        $\mathbf{m}_v \gets \mathbf{m}_v + \frac{\sigma}{f_u} \cdot \mathbf{m}_u$\; 
    }
    $\mathbf{m}_v[v] \gets 1$\;
    $\mathcal{C} \gets N(v) \cap \desc (v)$\; 
    $f_v \gets {d_{v}-\sum_{x \in \mathcal{C}} \mathbf{m}_v[x]}$\;
    
    \KwRet{$(\mathbf{m}_v,f_v)$}\;  
}

\SetKwFunction{FMain}{Main}
\SetKwProg{Fn}{Function}{:}{}
\Fn{\FMain{$\mathcal{G}$}}{
    $\mathcal{H} \gets \buildtree(\mathcal{G})$\;
    $(\mathbf{m}_v,f_v)$ for all $v\in\mathcal{V}$ $\gets \FSub(\mathcal{H}.\rooth)$\;
    \KwRet{Hierarchy tree $\mathcal{H}$, $(\mathbf{m}_v,f_v)$ for all $v\in\mathcal{V}$}\;
}
\end{algorithm}

\stitle{Index Building. }
Once the hierarchy is built, we perform a recursive traversal from the root.  
For each node $v$, if there exist nodes in $\desc (v)$ whose indices have not yet been computed, the algorithm first processes those descendants (Lines~2–4).  
After all indices of $\desc (v)$ are available, the index of $v$ is computed using Equation~\ref{eq:calmf}, obtaining the pair $(\mathbf{m}_v, f_v)$ (Lines~5–11).  
This procedure continues until all nodes have been processed, resulting in the complete index structure for the entire graph. This process correspond to Figure~\ref{fig:procedure4all}(c).
\begin{lemma}[Time Complexity of Index Building]
Let $h$ be the height of the hierarchy tree $\mathcal{H}$.  
The total running time of the index building algorithm (Algorithm~\ref{alg:indexbuilding}) is $O(n\cdot h\cdot (h+ d_{max}))$.
\end{lemma}

\begin{proof}
Lines~1--5 ensure that the procedure is invoked once for each vertex $v \in V$, thus contribute only $O(n)$ time. The dominant cost comes from Lines~6--9, whose total work over all nodes can be expressed as $\sum_{v \in V} \sum_{u \in \desc(v)} (\deg(u) + |\desc(u)|)$. For the first term, $\sum_{v \in V} \sum_{u \in \desc(v)} \deg(u)$, by reindexing over descendants we obtain $\sum_{v \in V} \sum_{u \in \desc(v)} \deg(u) = \sum_{u \in V} \deg(u) \cdot |\anc(u)|$, where $\anc(u)$ denotes the set of ancestors of $u$ in the hierarchy. Since the hierarchy tree has height $h$, we have $|\anc(u)| \le h$ for all $u$, and letting $d_{\max} = \max_{u \in V} \deg(u)$, this term is bounded by $O(n h d_{\max})$. For the second term, $\sum_{v \in V} \sum_{u \in \desc(v)} |\desc(u)|$, a similar reindexing yields $\sum_{v \in V} \sum_{u \in \desc(v)} |\desc(u)| = \sum_{w \in V} \sum_{v \in \anc(w)} \sum_{u \in \anc(v)} 1$. Each node has at most $h$ ancestors, so the inner double sum is $O(h^{2})$ for every $w$, and hence this term is bounded by $O(n h^{2})$. Combining the two parts, the overall time complexity is $O(n h d_{\max} + n h^{2})$.
\end{proof}
As discussed earlier, $h$ is typically small relative to $n$, thus the time required to build \bindex\ remains acceptable in practice. For example, on \fullusa, our index is constructed within 4.5 hours.

\subsection{Query processing}

\begin{algorithm}[t]
\small
\caption{Query processing algorithm with \bindex}
\label{alg:bd-query}
\KwIn{Hierarchy tree $\mathcal H$, $(\mathbf{m}_v,f_v)$ for all $v\in\mathcal{V}$, query pair $(s,t)$}
\KwOut{Biharmonic distance $b(s,t)$}
$\tilde{\boldsymbol{\tau}}_s[v],\tilde{\boldsymbol{\tau}}_t[v]\gets0$ for all $v\in\mathcal{V}$\;

\ForEach{$u \in \anc(s)$}{
  \ForEach{$k \in \desc(u)$ }{ 
    $\tilde{\boldsymbol{\tau}}_s[k]\gets \tilde{\boldsymbol{\tau}}_s[k] + \frac{\mathbf{m}_u[s]}{f_u}\cdot \mathbf{m}_u[k]$\; 
  }
}
\ForEach{$u \in \anc(t)$}{
  \ForEach{$k \in \desc(u)$ }{ 
    $\tilde{\boldsymbol{\tau}}_t[k]\gets \tilde{\boldsymbol{\tau}}_t[k] + \frac{\mathbf{m}_u[t]}{f_u}\cdot \mathbf{m}_u[k]$\; 
  }
}

\Return $b(s,t)=\|\tilde{\boldsymbol{\tau}}_{s}-\tilde{\boldsymbol{\tau}}_{t}\|_2^2 - \frac{1}{n}\,\Big(\mathbf 1^{\top}(\tilde{\boldsymbol{\tau}}_{s}-\tilde{\boldsymbol{\tau}}_{t}) \Big)^2$ according to Lemma~\ref{lem:v-absorbed}\;
\end{algorithm}

With the constructed \bindex, we can derive the contribution matrix $\widetilde{\mathbf{M}}_v$ for every node $v \in \mathcal{V}$.  
Based on these matrices, the deterministic random walk distribution $\tilde{\boldsymbol{\tau}}$ for any starting node $s$ can be directly computed.  

During query processing, given a query pair $(s, t)$, we only need to access the stored pairs $(\mathbf{m}, f)$ along the ancestor sets $\anc(s)$ and $\anc(t)$.  
By accumulating the stored information of these ancestors, we obtain the deterministic random walk distributions $\tilde{\boldsymbol{\tau}}_s$ and $\tilde{\boldsymbol{\tau}}_t$ for the two source nodes (corresponding to Lines~2–7 in Algorithm~\ref{alg:bd-query}).  
Finally, using the random walk formulation of the biharmonic distance in Lemma~\ref{lem:v-absorbed}, we compute the value $b(s,t)$ (Line~8). For example, in Figure~\ref{fig:procedure4all}(d), the query for $(v_2,v_5)$ needs access to \bindex on $\anc(v_2)\cup \anc(v_5) = \{v_2,v_5,v_6,v_7,v_8\}$ and does not access $\{v_0,v_1,v_3,v_4\}$.

\begin{lemma}[Time Complexity of Query Processing]
Let $h$ be the height of the hierarchy tree $\mathcal{H}$.  
For a query pair $(s,t)$, the time complexity of computing $b(s,t)$ using the index structure is $O(n\cdot h)$ in the worst case.
\end{lemma}

\begin{proof}
Algorithm~\ref{alg:bd-query} accesses only $\mathrm{Anc}(s)$ and $\mathrm{Anc}(t)$. For each $u$ on those two chains (each of length $\le h$) it loops over $\mathrm{Desc}(u)$ to update $\boldsymbol{\tau}_s$ or $\boldsymbol{\tau}_t$. In the worst case, the sum of $|\mathrm{Desc}(u)|$ over all $u$ on a chain is $O(n\cdot h)$. Therefore, the time complexity is $O(n\cdot h)$.
\end{proof}

Since each query traverses only a small portion of the hierarchy, most computations are localized to a few relevant subtrees. In practice, the query time is far below the worst-case bound. On \fullusa dataset, our average query time is less than $10^2$s and is almost exact, while other methods require more than $10^5$s. 

\stitle{Discussion: \bindex is an exact method.} Although we represent \bd using random walks, it is important to emphasize that \bindex is numerically exact up to floating-point precision. The random walk probabilities are computed deterministically in our approach, and no random walk sampling is involved at any stage. Consequently, \bindex fundamentally differs from previous projection-based \cite{ACC2018Yi} and random walk-based \cite{LiuKDD2024} approximate methods. The only source of numerical error arises from floating-point precision, which is inherent to all methods for computing \bd, since \bd is a numerical quantity. 
The experiments further demonstrate that \bindex is highly accurate: across all datasets, the relative error consistently remains below $10^{-9}$.

\section{Experiments}
\subsection{Experimental setup}
\label{sec:exp-setup}

\stitle{Datasets, query sets, and ground truth.}
All datasets used are listed in Table~\ref{tab:datasets}, which are publicly available from the \snap~\cite{Leskovec2014SNAP} and \dimacs~\cite{dataset-DIMACS}. For each dataset, we uniformly sample $100$ distinct source--target vertex pairs as the query set. Although \bindex is an exact method, we require ground-truth results to evaluate prior approximate methods for \bd \cite{ACC2018Yi,LiuKDD2024}.
For ground-truth \bd values, we exactly solve the underlying Laplacian linear systems using a Cholesky factorization of $\boldsymbol{L}$, i.e., a direct sparse linear solve~\cite{HornJohnson2013MatrixAnalysis}.

\stitle{Environment.}
All experiments are conducted on a Linux server with Intel Xeon E5\mbox{-}2680~v4 CPU and 512\,GB RAM.
Algorithms are implemented in C++ and compiled with \texttt{g++}~7.5.0. 

\stitle{Different algorithms.}
We compare \bindex with 2 \sota exact methods. Except for Cholesky factorization,  we also evaluate \lapsolver, a Laplacian solver that combines approximate Gaussian elimination with a preconditioned conjugate gradient (PCG) phase~\cite{Cohen2014SolvingSDD, Gao2023RobustLaplacian, KyngSachdeva2016, rchol}, with its accuracy parameter set to $10^{-15}$. We also compare against 5 \sota approximate methods for single-pair \bd query, including (i) \emph{Laplacian solver}-based method. \randomprojection~\cite{ACC2018Yi} uses JL-sketch that solves $O(\log n)$ Laplacian systems in preprocessing and answers each query by computing an inner product. (ii) \emph{Random walk}-based approaches. \stw~\cite{LiuKDD2024} is a sampled walk estimator that draws $r$ pairs of truncated random walks up to length $\ell$ to form an additive-accuracy estimate. \swf~\cite{LiuKDD2024} is a sampling-with-feedback variant that uses empirical-variance tests for early stopping. \pushp~\cite{LiuKDD2024} is a local push method on a truncated expansion with a pair-dependent truncation length (we report \pushp and do not run \push separately). Unless otherwise noted, the default $\epsilon=0.1$ follows the previous study~\cite{LiuKDD2024}. For \bindex, we use recursive minimum cut heuristic \cite{METIS} to build the hierarchy tree $\mathcal{H}$.
\begin{table}[t]
\centering
\caption{Datasets}
\label{tab:datasets}
\scalebox{0.86}{
\begin{tabular}{c c c c c}
\hline
\textbf{Dataset} & \textbf{$n$} & \textbf{$m$} & \textbf{$d_{max}$} & \textbf{Type} \\
\hline
\textsc{Facebook}     & 4,039      & 88,234      & 1,045  & Social \\
\textsc{CAIDA}        & 26,475     & 53,381      & 2,628  & Internet \\
\textsc{Email-Enron}  & 33,696     & 180,811     & 1,383  & Social \\
\textsc{NewYork}      & 264,346    & 365,050     & 8      & Road \\
\textsc{DBLP}         & 317,080    & 1,049,866   & 343    & Collaboration \\
\textsc{Amazon}       & 334,863    & 925,872     & 549    & Co-purchase \\
\textsc{Road-PA}      & 1,090,920  & 1,541,898   & 9      & Road \\
\textsc{Road-TX}      & 1,393,383  & 1,921,660   & 12     & Road \\
\textsc{Road-CA}      & 1,971,281  & 2,766,607   & 12     & Road \\
\textsc{Full-USA}     & 23,947,348 & 28,854,319  & 9      & Road \\
\hline
\end{tabular}}
\end{table}

\subsection{Query processing performance}
\label{sec:query-efficiency}
We evaluate the query processing performance of different algorithms on the same set of $100$ queries, reporting the average \emph{query time} for all methods and \emph{relative error} only for approximate methods. The comparison with exact solvers is reported in Fig.~\ref{fig:time-exact}, while the comparison with approximate methods is shown in Figs.~\ref{fig:time-approx} and~\ref{fig:error-approx}.

\stitle{Comparison with exact methods.}
Figure~\ref{fig:time-exact} reports the query time of \bindex against two exact baselines, \lapsolver and a Cholesky decomposition, across all datasets. \bindex is consistently $2$--$4$ orders of magnitude faster than both exact methods. For example, on the large road network \fullusa, a single \bd query requires about $1\times 10^{6}$ seconds with Cholesky and about $4\times 10^{5}$ seconds with \lapsolver, which is clearly impractical for single-pair queries, whereas our index answers the same query in only $64$ seconds. By constructing \bindex, our approach makes exact \bd queries for large numbers of single pairs on massive networks practically feasible.

Figure~\ref{fig:exact-error} reports the relative error of \bindex and \lapsolver when using the Cholesky factorization as the ground-truth baseline. As shown in the figure, the relative error of \bindex remains below $10^{-9}$ on all datasets. The residual error of our method stems solely from floating-point roundoff: the index is stored in double-precision floating-point format (\texttt{double} in C++), which provides about $15$ significant digits, and the local rounding errors introduced at each step of the hierarchy are propagated and accumulated, resulting in the observed $10^{-9}$-level global error.

\begin{figure}
    \centering
    \includegraphics[width=0.92\linewidth]{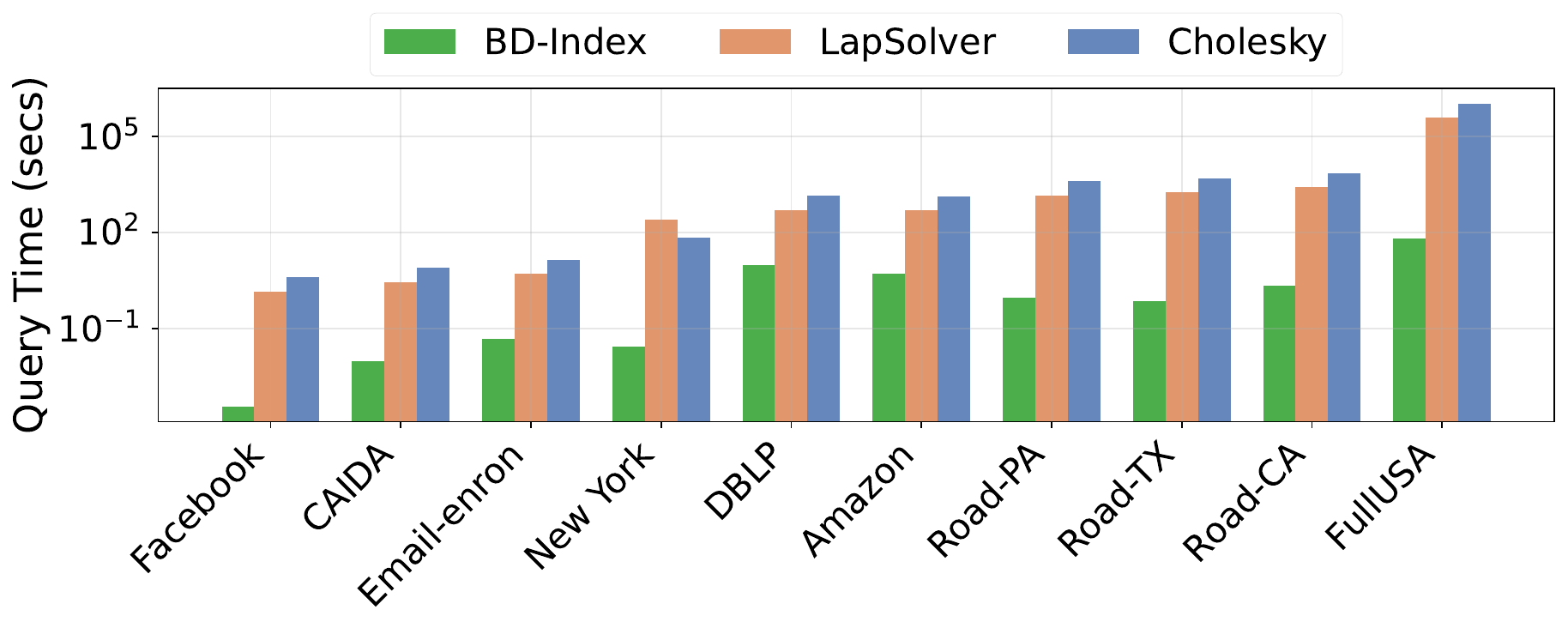}
    \caption{Query time compared with exact methods}
    \label{fig:time-exact}
\end{figure}

\begin{figure}
    \centering
    \includegraphics[width=0.92\linewidth]{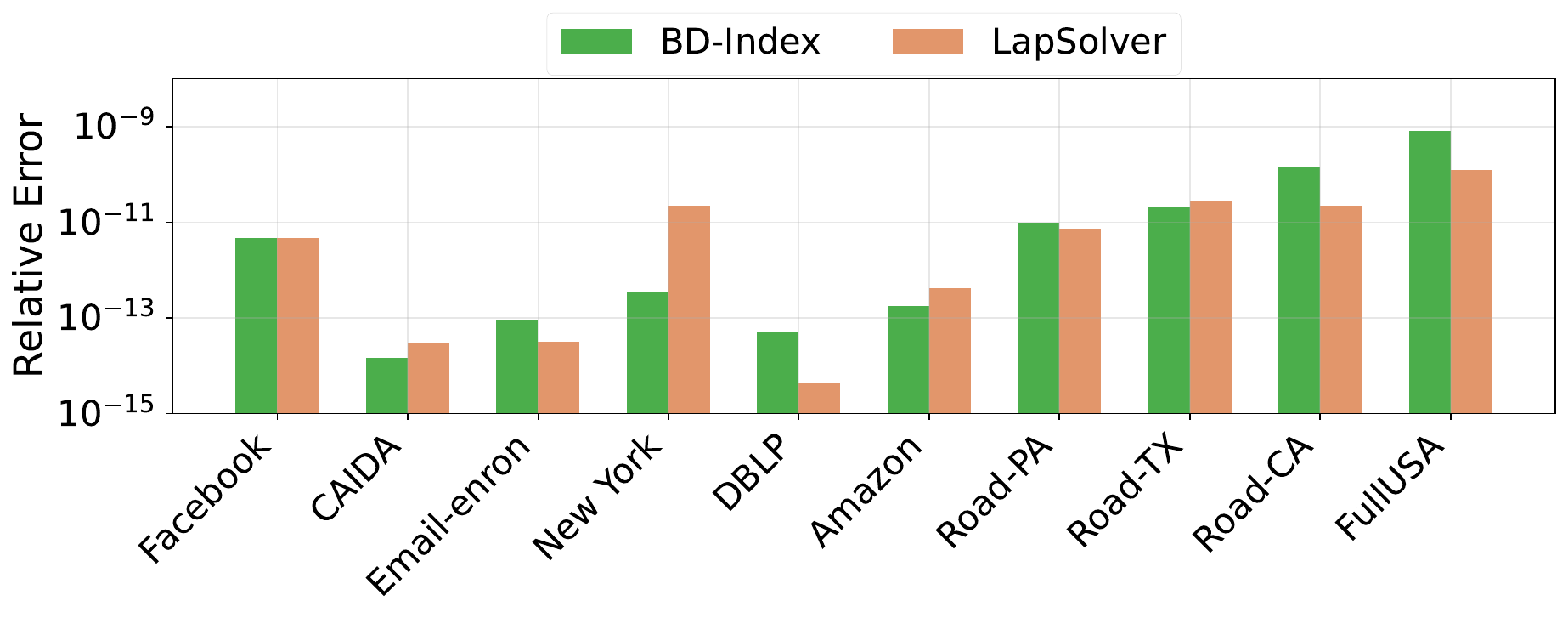}
    \caption{Relative error of \bindex and \lapsolver}
    \label{fig:exact-error}
\end{figure}

\stitle{Comparison with approximate methods.}
As shown in Figs.~\ref{fig:time-approx} and~\ref{fig:error-approx}, once the index has been constructed, \bindex, although computing exact results,  achieves at least an order-of-magnitude speedup over the fastest competing approximate method on all datasets. \randomprojection can only handle small graphs. \stw, \swf, and \pushp are competitive on social graphs but their performance deteriorates sharply on road networks, where \bindex is up to two additional orders of magnitude faster. For example, on the \textsc{Amazon} dataset, our method answers queries in about $6$ seconds with exact accuracy, whereas other methods require at least $10^{2}$ seconds to reach a relative error of $10^{-4}$. Overall, existing approximate methods incur very large errors on road networks and often fail to produce usable results, this is because they all rely on fast random walk mixing—precisely the bottleneck that our divide-and-conquer indexing strategy is easy to overcome.

\begin{figure*}
    \centering
    \includegraphics[width=0.95\linewidth]{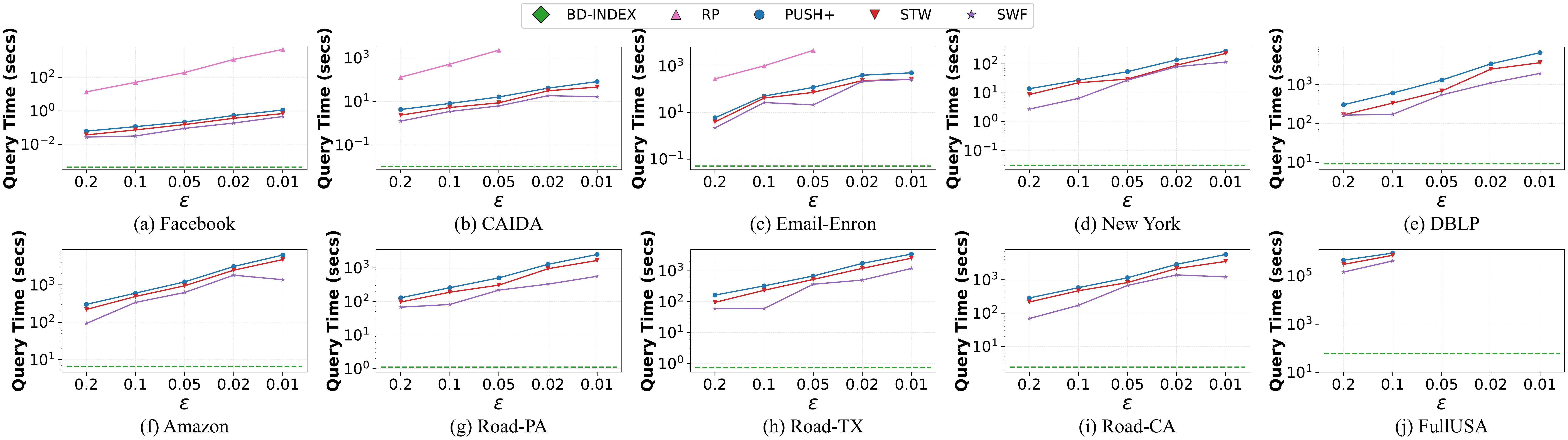}
    \caption{Query time compared with different approximate methods}
    \label{fig:time-approx}
\end{figure*}

\begin{figure*}
    \centering
    \includegraphics[width=0.95\linewidth]{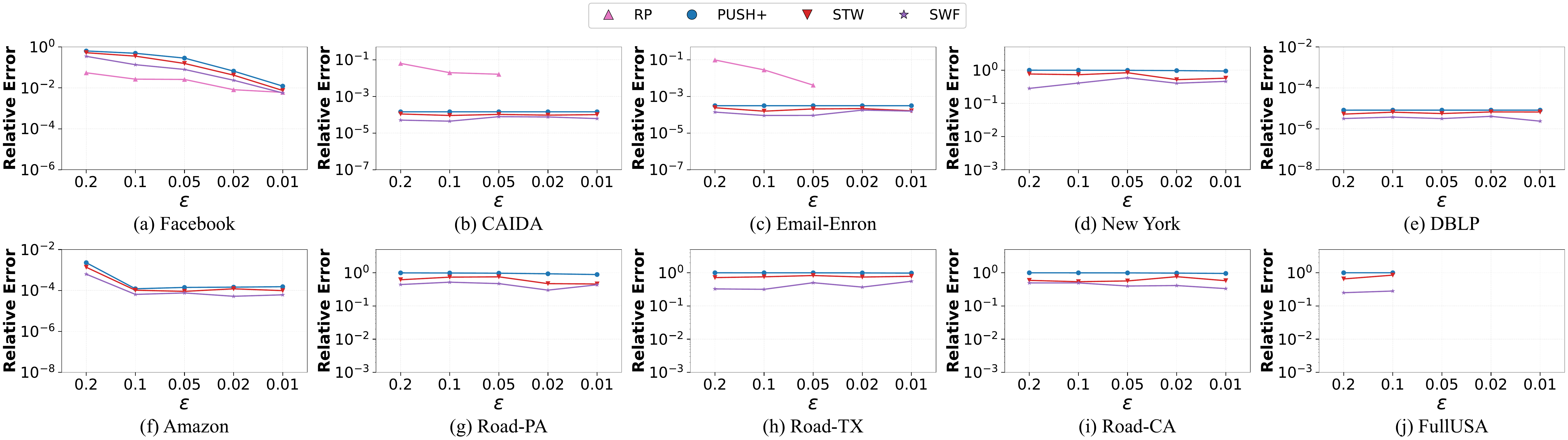}
    \caption{Query relative error of different approximate methods}
    \label{fig:error-approx}
\end{figure*}

\subsection{Index analysis}
The significant improvement of \bindex over existing approaches is due to the sacrifice of index building time and index size. In this experiment, we report the index building time and index size on 10 datasets. The results can be found in Table~\ref{tab:index_info}.

\stitle{Index size.} Recall that the index size can be bounded by $O(n\cdot h)$. As can be seen, the indices are compact when the hierarchy tree height $h$ is modest. As $h$ increases on social networks, the stored vectors become much larger. On large road networks, despite very large $n$, the indices remain within a single-machine budget when $h$ stays moderate. For example, the whole index of \fullusa takes $167,597$ MB with $h=1,622$, because \fullusa\ is a large road network.
For social graphs significantly larger than those in Table~\ref{tab:index_info}, the index size may exceed the available memory on our machine, thus the build cannot complete. In conclusion, our \bindex is highly efficient for large graphs with small $h$ (e.g., road networks). However, for graphs with large $h$, the $O(n\cdot h)$ storage cost remains a limitation, indicating a clear direction for future optimization.

\stitle{Index building time.}
The index building time follows the same pattern. Most road networks with moderate tree hierarchy $h$ complete index construction in minutes.
When $h$ is large, social/information graphs incur substantial offline time. For example, \textsf{Amazon} takes $78{,}957$ seconds and \textsf{DBLP} takes \(368{,}690\) seconds. However, \textsf{Road-PA} takes only $198$ seconds even though this graph is an order of magnitude larger than \textsf{Amazon} and \textsf{DBLP}. These results are consistent the \(n\cdot h\cdot (h+d_{max})\) dependence for the current implementation.

\bindex achieves theoretically exact answers and enables single-pair queries that are one to three orders of magnitude faster than the fastest approximate method, and two to four orders of magnitude faster than existing exact solvers, at the cost of an offline index whose memory footprint scales with $n$ and $h$. In practice, it is efficient on very large graphs with small $h$ (e.g., road networks) and on medium-scale graphs with large $h$ (e.g., social networks).

\begin{table}[t]
\centering
\caption{\textbf{Index Performance on all datasets}}
\label{tab:index_info}
\scalebox{0.82}{
\begin{tabular}{ccccc}
\toprule
\multirow{2}{*}{\textbf{Datasets}} & \textbf{Graph Size} & \textbf{Hierarchy} & \textbf{Index Size} & \textbf{Construction} \\
                          & (MB)       &        \textbf{tree height} $h$            & (MB)       & \textbf{Time} (secs) \\
\midrule
\textsc{Facebook}    & 0.8   & 401   & 5      & 0.55   \\
\textsc{CAIDA}   & 0.5  & 265  & 37  & 44  \\
\textsc{Email-enron} & 1.8   & 2,455  & 299    & 1,549   \\
\textsc{NewYork}     & 4.6   & 295   & 352    & 2.34   \\
\textsc{DBLP}        & 13    & 18,466 & 33,470  & 368,690 \\
\textsc{Amazon}      & 12    & 11,805 & 21,358  & 78,957  \\
\textsc{Road-PA}     & 21    & 817   & 5,496   & 198    \\
\textsc{Road-TX}     & 26    & 607   & 4,318   & 77     \\
\textsc{Road-CA}   & 40  & 857  & 9,318  & 253  \\
\textsc{FullUSA}     & 470   & 1,622  & 167,597 & 16,023  \\
\bottomrule
\end{tabular}}
\end{table}

\subsection{Comparison of different tree hierarchy $\mathcal{H}$}
In this experiment, we compare different tree hierarchy $\mathcal{H}$ to cut the graph into pieces. Since both the time and space complexity of our algorithm depend on the tree height $h$, the structure of the tree hierarchy $\mathcal{H}$ has a direct impact on the overall efficiency. In Section~\ref{sec:subsecib}, we introduced two heuristic strategies for constructing $\mathcal{H}$: \emph{minimum cut-based} hierarchy construction and \emph{minimum degree-based} hierarchy construction. Figures~\ref{fig:separator-query}, \ref{fig:separator-build}, and~\ref{fig:separator-size} report the query time, index construction time, and index size, respectively, under these two hierarchy constructions across all datasets. Across the board, the minimum cut-based hierarchy yields shorter query time, lower construction time, and a smaller index on most datasets.

Table~\ref{tab:hier-nnz} further compares the structural properties of the resulting hierarchies in terms of the tree height $h$ and average label size $s$, which is the average number of non-zero elements stored in our index. The minimum cut-based strategy typically produces hierarchies with smaller tree height $h$ and lower average label size $s$. This structurally more compact hierarchy directly translates into the superior query and indexing performance observed in our experiments.

\begin{figure}
    \centering
    \includegraphics[width=0.92\linewidth]{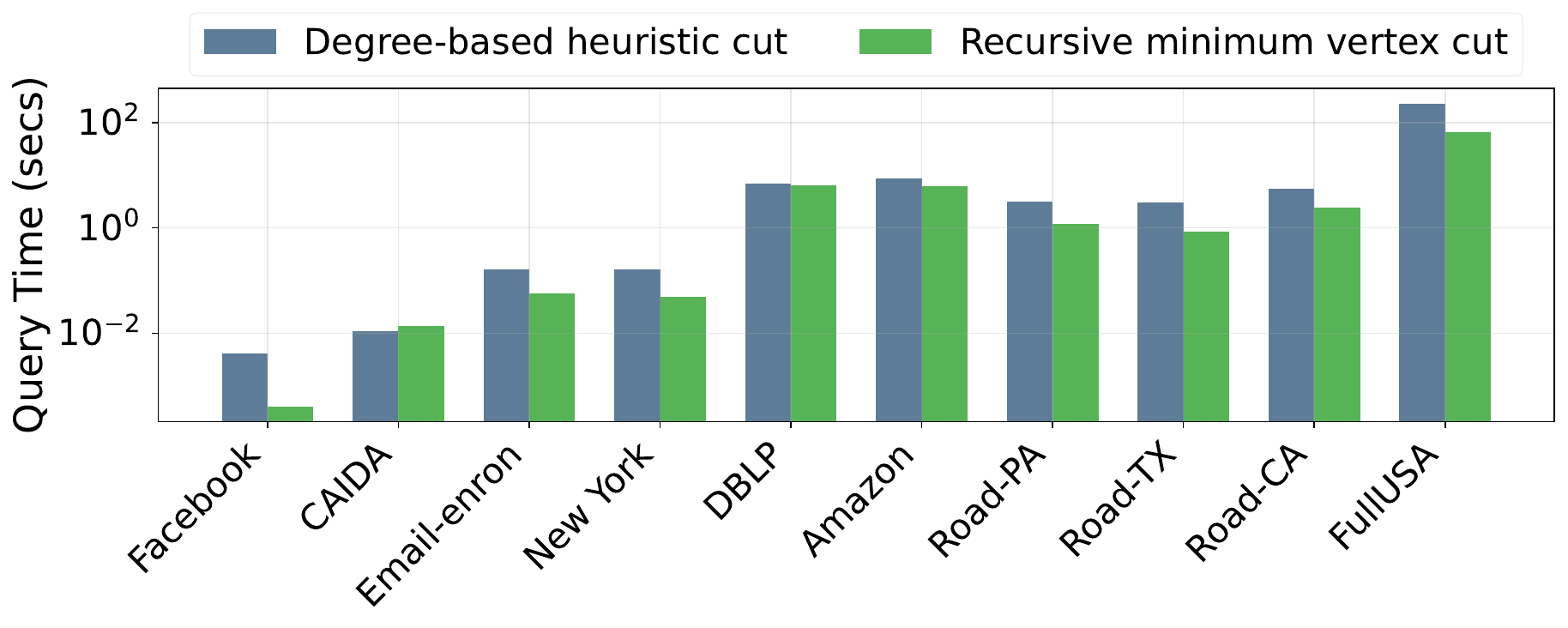}
    \caption{Query time with different tree hierarchy $\mathcal{H}$}
    \label{fig:separator-query}
\end{figure}
\begin{figure}
    \centering
    \includegraphics[width=0.92\linewidth]{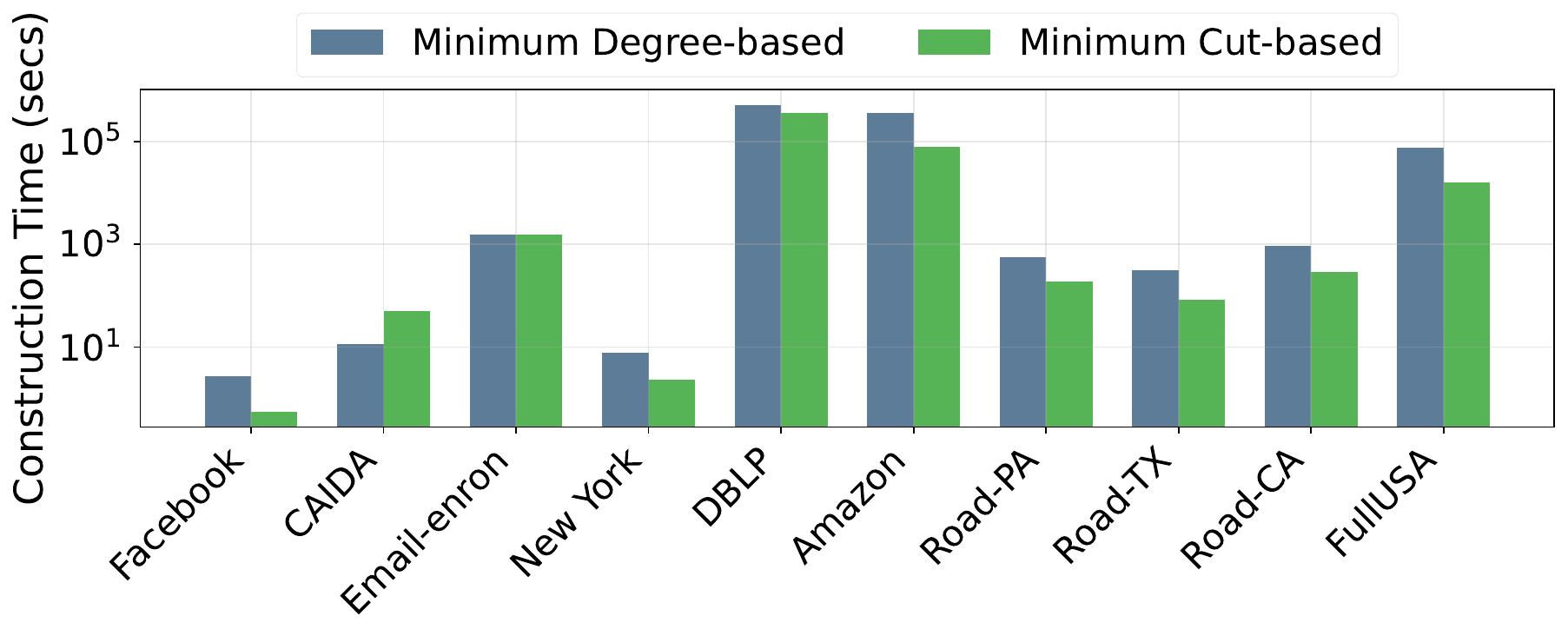}
    \caption{Index building time with different tree hierarchy $\mathcal{H}$}
    \label{fig:separator-build}
\end{figure}
\begin{figure}
    \centering
    \includegraphics[width=0.92\linewidth]{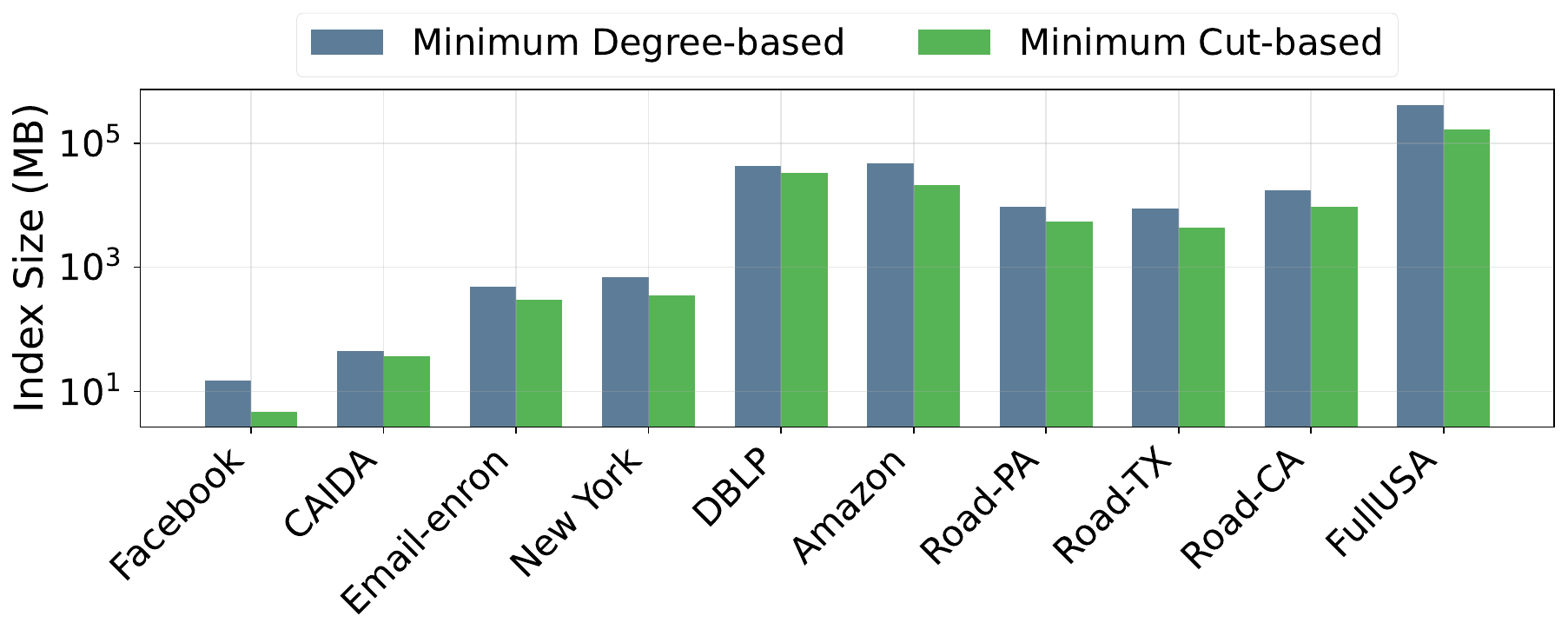}
    \caption{Index size with different tree hierarchy $\mathcal{H}$}
    \label{fig:separator-size}
\end{figure}

\begin{table}[t]
\centering
\caption{Comparison of hierarchy tree height $h$ and average label size $s$ under different tree hierarchy $\mathcal{H}$
}
\label{tab:hier-nnz}
\scalebox{0.82}{%
\begin{tabular}{ccccc}
\hline
\multirow{2}{*}{\textbf{Dataset}} &
\multicolumn{2}{c}{\textbf{Hierarchy tree height $h$}} &
\multicolumn{2}{c}{\textbf{Average label size $s$}} \\
[0.3ex]
\cline{2-5}
\noalign{\vskip 0.4ex}
& \shortstack{Minimum\\cut-based}
& \shortstack{Minimum\\degree-based}
& \shortstack{Minimum\\cut-based}
& \shortstack{Minimum\\degree-based} \\
\hline
Facebook    & 401    & 741    & 154    & 486    \\
CAIDA       & 265    & 289    & 181    & 222    \\
Email-Enron & 2,455  & 2,397  & 1,166  & 1,895  \\
New York    & 295    & 767    & 174    & 346    \\
DBLP        & 18,466 & 20,109 & 13,836 & 17,582 \\
Amazon      & 11,805 & 21,394 & 8,360  & 18,770 \\
Road-PA     & 817    & 2,034  & 662    & 1,146  \\
Road-TX     & 607    & 1,530  & 419    & 867    \\
Road-CA     & 857    & 1,715  & 624    & 1,167  \\
FullUSA     & 1,622  & 3,976  & 917    & 2,268  \\
\hline
\end{tabular}%
}
\end{table}

\section{Case studies}
In this section, we conduct two case studies to demonstrate the effectiveness of \bd in two graph mining tasks.
\subsection{Critical link identification on road networks}
\label{sec:app-bd-roads}

We apply \bd to identify \emph{critical links} in a road network~\cite{IJCAI2018Yi}. 
For each edge $(s,t)$, $b(s,t)$ acts as an \emph{edge centrality} measure: larger $b(s,t)$ indicate links that are more critical for preserving network connectivity. We evaluate this interpretation on a $3{\times}3$\,km road network of Philadelphia extracted from OpenStreetMap, which contains real urban infrastructure such as highways, local streets, and bridges. We keep the largest connected component and treat bridges and river-crossing segments as weak ground-truth critical links. We compare \bd\ against three alternative edge centralities: \er\ (effective resistance using inverse road length), \eb\ (edge betweenness with road length as impedance), and \pr\ (an edge score derived from node PageRank computed on inverse length). In Fig.~\ref{fig:color}, edges are visualized from blue (low centrality) to red (high centrality). \bd\ highlights all bridge segments while remaining low on internal roads; \er\ also highlights bridges but assigns similarly high scores to many adjacent streets, showing limited contrast; and both \eb\ and \pr\ fail to consistently detect all bridges. Overall, \bd\ most cleanly isolates the visually critical links in the network.

To further quantify this effect, we iteratively remove top-ranked edges under each centrality until the total removed length reaches $5\%$ or $10\%$ of the network, and measure the impact on connectivity using three metrics: \textsf{LCC}, the fraction of nodes in the largest connected component; \textsf{\#Comp}, the number of connected components; and \textsf{Reach}, the fraction of reachable origin–destination pairs among $10^3$ randomly sampled pairs. Lower values of \textsf{LCC} and \textsf{Reach}, together with higher \textsf{\#Comp}, indicate stronger disruption. As reported in Table~\ref{tab:conn-length}, under both removal budgets, \bd\ yields the lowest \textsf{LCC} and \textsf{Reach} and the highest \textsf{\#Comp}, confirming that it focuses on true structural bottlenecks. \pr\ performs second best, while \er\ and \eb\ have the weakest impact on network connectivity. These results demonstrate the effectiveness of \bd\ for identifying critical links in real road networks.

\begin{figure}
    \centering
    \includegraphics[width=0.85\linewidth]{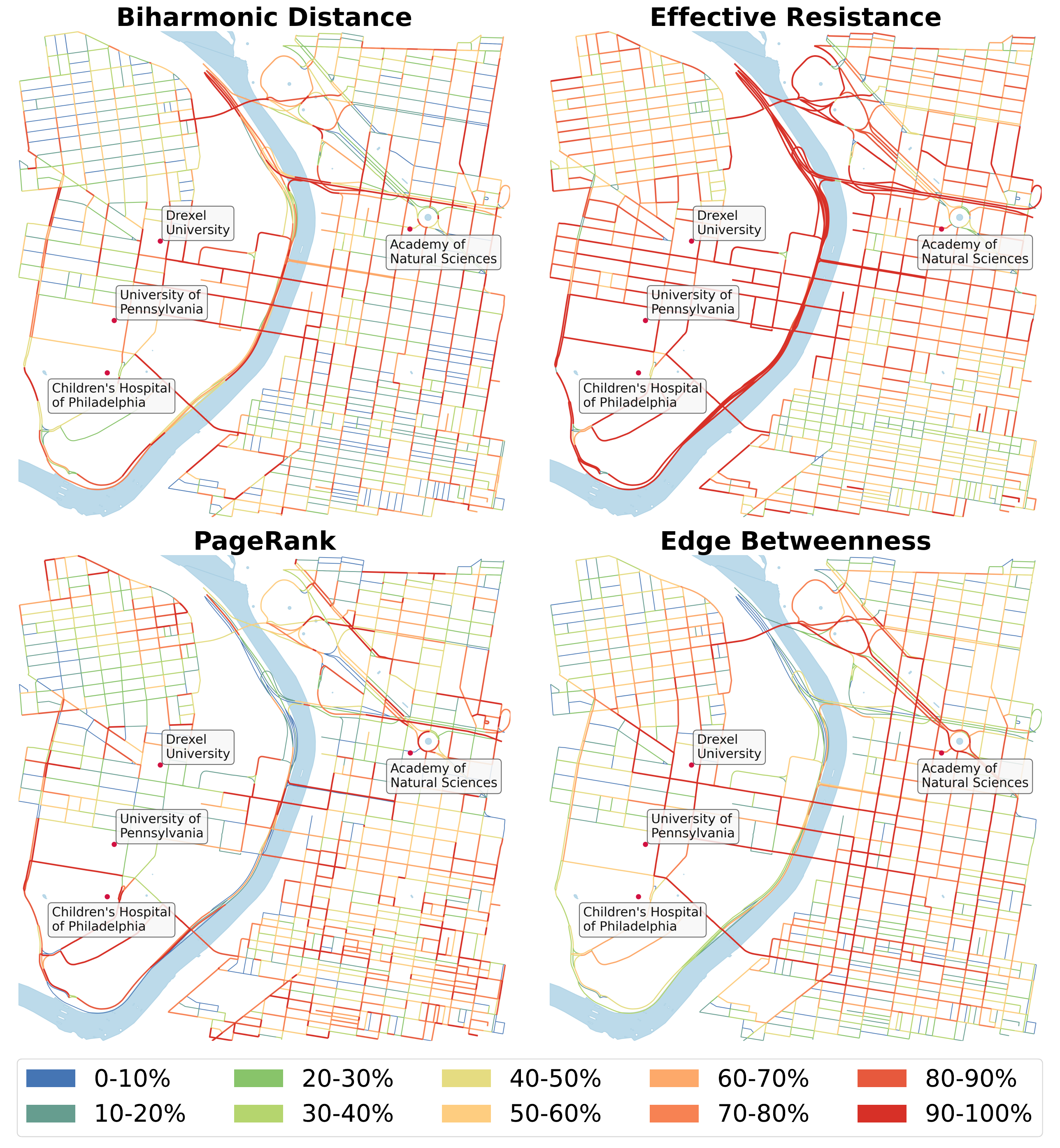}
    \caption{Comparison of different edge centralities on the Philadelphia network, edges are marked with colors from blue (low centrality value) to red (high centrality value).}
    \label{fig:color}
\end{figure}

\begin{table}[t]
\centering
\caption{Comparison of network connectivity (LCC, \#Comp, Reach) after removing the top 5\% / 10\% of road length ranked by different edge centrality measures}
\label{tab:conn-length}
\small
\scalebox{0.9}{
\begin{tabular}{cccccc}
\toprule
\textbf{Budget} & \textbf{Method} & \textbf{LCC} & \textbf{\#Comp} & \textbf{Reach} \\
\midrule
\multirow{4}{*}{5\%}
& \bd      & \textbf{0.675} & 12 & \textbf{0.528} \\
& \er      & 0.994 & 6  & 0.991 \\
& \eb      & 0.993 & 4  & 0.983 \\
& \pr      & 0.988 & 12 & 0.975 \\
\midrule
\multirow{4}{*}{10\%}
& \bd      & \textbf{0.664} & 29 & \textbf{0.481} \\
& \er      & 0.986 & 14 & 0.981 \\
& \eb      & 0.975 & 12 & 0.944 \\
& \pr      & 0.971 & 27 & 0.940 \\
\bottomrule
\end{tabular}}
\end{table}

\subsection{Over-squashing mitigation in \gnn}
\label{sec:app-bd-rewire}
Graph Neural Networks (\gnn{}s) often suffer from \emph{over-squashing}~\cite{ICLR2022Jake,ICML2023Black,AKANSHA2025130389}, where information from many distant nodes is forced through a few narrow connections, weakening message passing and degrading accuracy. Prior work~\cite{ICML2023Black} has shown that \bd\ can be used to identify and alleviate such bottlenecks via graph rewiring, thereby mitigating over-squashing and improving \gnn{} performance. Here, we perform a simple replication-style study on the \textsc{Chameleon}~\cite{chameleon2021} dataset (a standard benchmark constructed from Wikipedia pages on diverse topics), applying the same \bd-guided rewiring strategy as in~\cite{ICML2023Black}. Our goal is not to design a new \gnn{} architecture, but to validate the role of \bd\ in reducing over-squashing.

A common theoretical indicator of over-squashing is the graph’s \emph{total effective resistance} (TotalER): higher TotalER implies more constrained information flow. We therefore use TotalER as a measure for the severity of over-squashing. Starting from the original graph (Baseline), we iteratively add a small number of new edges in up to three rounds. In each round, we select node pairs with the largest \bd\ values from a pool of long-range pairs and add $1\%$ new edges. After each step, we record both TotalER and test accuracy.

As shown in Table~\ref{tab:bd-rewire-chameleon}, adding only $1\%$ of \bd-guided edges already reduces TotalER by about $3.9\%$ while increasing accuracy by roughly $0.35$ percentage points. Increasing the addition to $2$–$3\%$ of \bd-guided edges further decreases TotalER by up to $4.5\%$ and yields accuracy gains of around $0.70$–$0.75$ percentage points overall. The absolute improvement is modest because our setup deliberately uses a very simple two-layer \gnn{} with basic training (rather than sophisticated architectures or heavy hyperparameter tuning), and \gnn{} design is not the main focus of this paper. Nevertheless, these results demonstrate that \bd\ provides an effective signal for rewiring to mitigate over-squashing: selecting a small fraction of additional edges according to \bd\ already relieves communication bottlenecks (lower TotalER) and improves predictive performance.

\begin{table}[t] 
\centering
\caption{\bd-guided over-squashing mitigation in \gnn{}s on \textsc{Chameleon}~\cite{chameleon2021} node classification}
\label{tab:bd-rewire-chameleon}
\small
\scalebox{0.85}{
\begin{tabular}{ccccccc}
    \toprule
    \textbf{Graph} & \textbf{\#Edges} & \textbf{\begin{tabular}[c]{@{}c@{}}$\Delta$\#Edges\\ (\%)\end{tabular} } &
    \textbf{\begin{tabular}[c]{@{}c@{}}TotalER\\ ($\times 10^6$)\end{tabular}} & \textbf{\begin{tabular}[c]{@{}c@{}}$\Delta$TotalER \\ (\%)\end{tabular}} &
    \textbf{Acc.} & \textbf{$\Delta$Acc (\%)} \\
    \midrule
    Baseline 
      & 23370 & +0.0 & 7.52 & {+}0.00 & 62.28 & {+}0.00 \\
    \midrule
    \multirow{3}{*}{\begin{tabular}[c]{@{}c@{}}\bd-guided\\rewired\end{tabular}}
      & 23603 & +1.0 & 7.23 & {-}3.93 & 62.50 & {+}0.35 \\
      & 23837 & +2.0 & 7.22 & {-}4.06 & 62.72 & {+}0.70 \\
      & 24071 & +3.0 & 7.18 & {-}4.47 & 62.75 & {+}0.75 \\
    
\bottomrule
\end{tabular}}
\end{table}

\section{Related work} \label{sec:related}

Several works have examined the mathematical properties of the Laplacian pseudoinverse~\cite{DoyleSnell1984RandomWalks,Klein1993ResistanceDistance,Chen2007ResistanceNormalizedLaplacian,YangKlein2013RecursionResistance,Bapat2014GraphsMatrices,HornJohnson2013MatrixAnalysis}.
Early studies revealed its close connection to the \emph{resistance distance} (\er)~\cite{Klein1993ResistanceDistance}
and the theory of random walks on graphs~\cite{DoyleSnell1984RandomWalks,Levin2017MarkovChains}.
Later analyses explored its spectral interpretation,
linking resistance measures to Laplacian eigenvalues and network coherence~\cite{Chen2007ResistanceNormalizedLaplacian,YangKlein2013RecursionResistance,Tyloo2018RobustnessSynchrony,Tyloo2019KeyPlayerResistance,Bamieh2012CoherenceLargeScale}.
Quantities derived from the Laplacian pseudoinverse, including the effective resistance (\er), personalized PageRank (\pr), and biharmonic distance (\bd),
are widely used in graph learning, network analysis, and spectral algorithms~\cite{Lu2011LinkPredictionSurvey,Tsitsulin2020SLaQ,Black2024BiharmonicHigherOrder,Shimada2016GraphDistanceComplex,Tran2019SignedDistanceRecommender,Yang2022BipartiteSimilaritySearch}.
\er\ measures connectivity and robustness; its total value reflects how strongly a network is connected,
and has been used for robustness optimization, sparsification, and resistance-based embeddings~\cite{spielman2008sparsification,predari2023greedy,Liao2024,Liao-ER-Index,ICML2023Black,LiSachdeva2023EffectiveResistance}.
\pr\ quantifies the influence of a node relative to a source and underlies many ranking, recommendation, and diffusion methods~\cite{mahoney2012local,shur2023pagerank_embedding,xie2015edgeppr,Haveliwala2003SecondEigenvalueGoogle}.
\bd\ has been applied in clustering, centrality, and geometric processing,
with recent studies demonstrating its advantages over \er~\cite{Black2024BiharmonicHigherOrder,LiuKDD2024,IJCAI2018Yi,ACC2018Yi,TIT2022Yi,TOG2010Lipman,Jin2019ForestDistanceCloseness}.
Physical and dynamical systems perspectives linked these quantities
to synchronization and control in large networks~\cite{Tyloo2018RobustnessSynchrony,Tyloo2019KeyPlayerResistance}.
These studies focus on applications, whereas our work emphasizes computation.

Many recent studies employ random walk–based methods to compute quantities derived from the Laplacian pseudoinverse~\cite{Levin2017MarkovChains,FOCS2018Bressan,NeurIPS2013Lee,ICDM2006Tong,NeurIPS2015Banerjee}.
The computation of \er\ can be interpreted as counting the expected number of times a random walk starting at one node visits another~\cite{KDD2021Peng,Liao2024,Liao-ER-Index},
while \pr\ corresponds to a random walk with probabilities~\cite{Haveliwala2003SecondEigenvalueGoogle,KDD2017Wang,Gleich2015SIREV}.
These works have developed efficient random walk algorithms for such quantities.
Beyond purely probabilistic estimators, spectral and randomized projection techniques
such as the Johnson–Lindenstrauss transform~\cite{Achlioptas2003DatabaseFriendlyRand,MotwaniRaghavan1995RandomizedAlgorithms}
and nearly-linear Laplacian solvers~\cite{Cohen2014SolvingSDD,SpielmanTeng2014NearlyLinearSDD,Gao2023RobustLaplacian}
have also been adopted.
However, such methods cannot be directly applied to the computation of \bd,
since \bd\ measures the distance between the distributions of two random walks.
This form makes the problem substantially more challenging than \er\ or \pr\ computations.

Another related direction studies index-based shortest-path distance querying.  Beyond classical online algorithms such as Dijkstra, bidirectional search, and A* search~\cite{dechampeaux1977bidirectional,goldberg2005computing}, many methods precompute hopsets and distance labelings so that, after adding a sparse set of shortcuts, all pairs are reachable within a bounded number of hops while (approximate) distances are preserved~\cite{cohen1994polylog,gavoille2001distance,elkin2016hopsets}.  This leads to distance oracles and 2-hop labels with strong guarantees, including bounds on general graphs and on graphs with small treewidth, highway dimension, or skeleton dimension~\cite{cohen2002reachability,farzan2011compact,li2020scaling,abraham2010highway,kosowski2017beyond,alstrup2016sublinear,gupta2019exploiting,cohenaddad2017fast}.  Theoretical advances have been instantiated in practical systems such as TEDI~\cite{wei2010tedi}, pruned landmark and independent-set–based labelings~\cite{akiba2013fast,fu2013islabel}, hop-doubling and hierarchical 2-hop schemes~\cite{chang2012exact,jiang2014hop,ouyang2018when}, projected vertex separators and hierarchical cut labellings~\cite{chen2021p2h,farhan2023hierarchical,farhan2025dual}, as well as dynamic and learning-based indexes~\cite{ouyang2020efficient,zhang2022relative,koehler2025stable,zheng2023reinforcement}.  Our goal is also to build an index for fast pairwise distance queries, but \bd\ is a Laplacian-based distance between random-walk distributions.  It does not satisfy the cut and bounded-hop properties that these indexes exploit, so hopset and labeling techniques cannot be reused; we instead design new hierarchical decompositions and labels tailored to encoding biharmonic distances via the Laplacian pseudoinverse.

\section{Conclusion}
In this work, we revisited the biharmonic distance (\bd) from an algorithmic perspective and addressed the challenge of efficient single-pair \bd\ queries on large graphs. We introduced \bindex, a divide-and-conquer index structure that leverages small graph separators to decompose \bd\ computation into independent local problems. The index can be built in $O(n\cdot h\cdot (h+ d_{max}))$ time, requires $O(n\cdot h)$ space, and answers each query in $O(n\cdot h)$ time, where $h$ is the height of the hierarchy tree---typically much smaller than $n$. Extensive experiments on large-scale datasets demonstrate that \bindex\ achieves order-of-magnitude speedups over state-of-the-art approximate methods while achieving theoretically exact accuracy. Beyond efficiency, we also showcased the utility of \bd\ in downstream applications, including identifying critical links in road networks and mitigating the over-squashing problem in graph neural networks.

\balance
\bibliographystyle{ACM-Reference-Format}
\bibliography{CholWilson}


\begin{thebibliography}{84}


\ifx \showCODEN    \undefined \def \showCODEN     #1{\unskip}     \fi
\ifx \showDOI      \undefined \def \showDOI       #1{#1}\fi
\ifx \showISBNx    \undefined \def \showISBNx     #1{\unskip}     \fi
\ifx \showISBNxiii \undefined \def \showISBNxiii  #1{\unskip}     \fi
\ifx \showISSN     \undefined \def \showISSN      #1{\unskip}     \fi
\ifx \showLCCN     \undefined \def \showLCCN      #1{\unskip}     \fi
\ifx \shownote     \undefined \def \shownote      #1{#1}          \fi
\ifx \showarticletitle \undefined \def \showarticletitle #1{#1}   \fi
\ifx \showURL      \undefined \def \showURL       {\relax}        \fi
\providecommand\bibfield[2]{#2}
\providecommand\bibinfo[2]{#2}
\providecommand\natexlab[1]{#1}
\providecommand\showeprint[2][]{arXiv:#2}

\bibitem[\protect\citeauthoryear{Abraham, Fiat, Goldberg, and Werneck}{Abraham et~al\mbox{.}}{2010}]%
        {abraham2010highway}
\bibfield{author}{\bibinfo{person}{Ittai Abraham}, \bibinfo{person}{Amos Fiat}, \bibinfo{person}{Andrew~V. Goldberg}, {and} \bibinfo{person}{Renato~F. Werneck}.} \bibinfo{year}{2010}\natexlab{}.
\newblock \showarticletitle{Highway Dimension, Shortest Paths, and Provably Efficient Algorithms}. In \bibinfo{booktitle}{\emph{Proceedings of the 21st Annual {ACM-SIAM} Symposium on Discrete Algorithms (SODA)}}. \bibinfo{pages}{782--793}.
\newblock


\bibitem[\protect\citeauthoryear{Achlioptas}{Achlioptas}{2003}]%
        {Achlioptas2003DatabaseFriendlyRand}
\bibfield{author}{\bibinfo{person}{Dimitris Achlioptas}.} \bibinfo{year}{2003}\natexlab{}.
\newblock \showarticletitle{Database-Friendly Random Projections: Johnson--Lindenstrauss with Binary Coins}.
\newblock \bibinfo{journal}{\emph{J. Comput. System Sci.}} \bibinfo{volume}{66}, \bibinfo{number}{4} (\bibinfo{year}{2003}), \bibinfo{pages}{671--687}.
\newblock


\bibitem[\protect\citeauthoryear{Akansha}{Akansha}{2025}]%
        {AKANSHA2025130389}
\bibfield{author}{\bibinfo{person}{S. Akansha}.} \bibinfo{year}{2025}\natexlab{}.
\newblock \showarticletitle{Over-squashing in Graph Neural Networks: A comprehensive survey}.
\newblock \bibinfo{journal}{\emph{Neurocomputing}}  \bibinfo{volume}{642} (\bibinfo{year}{2025}), \bibinfo{pages}{130389}.
\newblock
\showISSN{0925-2312}


\bibitem[\protect\citeauthoryear{Akiba, Iwata, and Yoshida}{Akiba et~al\mbox{.}}{2013}]%
        {akiba2013fast}
\bibfield{author}{\bibinfo{person}{Takuya Akiba}, \bibinfo{person}{Yoichi Iwata}, {and} \bibinfo{person}{Yuichi Yoshida}.} \bibinfo{year}{2013}\natexlab{}.
\newblock \showarticletitle{Fast Exact Shortest-Path Distance Queries on Large Networks by Pruned Landmark Labeling}. In \bibinfo{booktitle}{\emph{Proceedings of the 2013 {ACM} {SIGMOD} International Conference on Management of Data}}. \bibinfo{pages}{349--360}.
\newblock


\bibitem[\protect\citeauthoryear{Alstrup, Dahlgaard, Knudsen, and Porat}{Alstrup et~al\mbox{.}}{2016}]%
        {alstrup2016sublinear}
\bibfield{author}{\bibinfo{person}{Stephen Alstrup}, \bibinfo{person}{S{\o}ren Dahlgaard}, \bibinfo{person}{Mathias B{\ae}k~Tejs Knudsen}, {and} \bibinfo{person}{Ely Porat}.} \bibinfo{year}{2016}\natexlab{}.
\newblock \showarticletitle{Sublinear Distance Labeling}. In \bibinfo{booktitle}{\emph{Proceedings of the 24th Annual European Symposium on Algorithms (ESA)}}.
\newblock


\bibitem[\protect\citeauthoryear{Bamieh, Jovanovi{\'c}, Mitra, and Patterson}{Bamieh et~al\mbox{.}}{2012}]%
        {Bamieh2012CoherenceLargeScale}
\bibfield{author}{\bibinfo{person}{B. Bamieh}, \bibinfo{person}{M. Jovanovi{\'c}}, \bibinfo{person}{P. Mitra}, {and} \bibinfo{person}{S. Patterson}.} \bibinfo{year}{2012}\natexlab{}.
\newblock \showarticletitle{Coherence in Large-Scale Networks: Dimension-Dependent Limitations of Local Feedback}.
\newblock \bibinfo{journal}{\emph{IEEE Trans. Automat. Control}} \bibinfo{volume}{57}, \bibinfo{number}{9} (\bibinfo{year}{2012}), \bibinfo{pages}{2235--2249}.
\newblock


\bibitem[\protect\citeauthoryear{Banerjee and Lofgren}{Banerjee and Lofgren}{2015}]%
        {NeurIPS2015Banerjee}
\bibfield{author}{\bibinfo{person}{Siddhartha Banerjee} {and} \bibinfo{person}{Peter Lofgren}.} \bibinfo{year}{2015}\natexlab{}.
\newblock \showarticletitle{Fast bidirectional probability estimation in Markov models}. In \bibinfo{booktitle}{\emph{Advances in Neural Information Processing Systems (NeurIPS)}}, Vol.~\bibinfo{volume}{28}. \bibinfo{pages}{1423--1431}.
\newblock


\bibitem[\protect\citeauthoryear{Bapat}{Bapat}{2014}]%
        {Bapat2014GraphsMatrices}
\bibfield{author}{\bibinfo{person}{R.~B. Bapat}.} \bibinfo{year}{2014}\natexlab{}.
\newblock \bibinfo{booktitle}{\emph{Graphs and Matrices (2nd Edition)}}.
\newblock \bibinfo{publisher}{Springer}.
\newblock


\bibitem[\protect\citeauthoryear{Black, D{\"o}rfler, and Gratton}{Black et~al\mbox{.}}{2024a}]%
        {Black2024BiharmonicHigherOrder}
\bibfield{author}{\bibinfo{person}{Matthew Black}, \bibinfo{person}{Florian D{\"o}rfler}, {and} \bibinfo{person}{Claudia Gratton}.} \bibinfo{year}{2024}\natexlab{a}.
\newblock \showarticletitle{Biharmonic Distance of Graphs and Its Higher-Order Variants: Analytical Properties with Applications to Centrality and Clustering}. In \bibinfo{booktitle}{\emph{Proceedings of the 2024 International Conference on Complex Networks}}.
\newblock


\bibitem[\protect\citeauthoryear{Black, Lin, Wong, and Nayyeri}{Black et~al\mbox{.}}{2024b}]%
        {Black2024}
\bibfield{author}{\bibinfo{person}{Mitchell Black}, \bibinfo{person}{Lucy Lin}, \bibinfo{person}{Weng-Keen Wong}, {and} \bibinfo{person}{Amir Nayyeri}.} \bibinfo{year}{2024}\natexlab{b}.
\newblock \showarticletitle{Biharmonic distance of graphs and its higher-order variants: theoretical properties with applications to centrality and clustering}. In \bibinfo{booktitle}{\emph{Proceedings of the 41st International Conference on Machine Learning}} \emph{(\bibinfo{series}{ICML'24})}. Article \bibinfo{articleno}{165}, \bibinfo{numpages}{26}~pages.
\newblock


\bibitem[\protect\citeauthoryear{Black, Wan, Nayyeri, and Wang}{Black et~al\mbox{.}}{2023}]%
        {ICML2023Black}
\bibfield{author}{\bibinfo{person}{Mitchell Black}, \bibinfo{person}{Zhengchao Wan}, \bibinfo{person}{Amir Nayyeri}, {and} \bibinfo{person}{Yusu Wang}.} \bibinfo{year}{2023}\natexlab{}.
\newblock \showarticletitle{Understanding oversquashing in GNNs through the lens of effective resistance}. In \bibinfo{booktitle}{\emph{Proceedings of the 40th International Conference on Machine Learning (ICML)}}. \bibinfo{publisher}{PMLR}, \bibinfo{pages}{2528--2547}.
\newblock


\bibitem[\protect\citeauthoryear{Bressan, Peserico, and Pretto}{Bressan et~al\mbox{.}}{2018}]%
        {FOCS2018Bressan}
\bibfield{author}{\bibinfo{person}{Marco Bressan}, \bibinfo{person}{Enoch Peserico}, {and} \bibinfo{person}{Luca Pretto}.} \bibinfo{year}{2018}\natexlab{}.
\newblock \showarticletitle{Sublinear algorithms for local graph centrality estimation}. In \bibinfo{booktitle}{\emph{59th IEEE Annual Symposium on Foundations of Computer Science (FOCS)}}. \bibinfo{publisher}{IEEE}, \bibinfo{pages}{709--718}.
\newblock


\bibitem[\protect\citeauthoryear{Bui and Jones}{Bui and Jones}{1992}]%
        {np}
\bibfield{author}{\bibinfo{person}{Thang Bui} {and} \bibinfo{person}{Curt Jones}.} \bibinfo{year}{1992}\natexlab{}.
\newblock \showarticletitle{Finding good approximate vertex and edge partitions is NP-hard}.
\newblock \bibinfo{journal}{\emph{Inform. Process. Lett.}} \bibinfo{volume}{42}, \bibinfo{number}{3} (\bibinfo{year}{1992}), \bibinfo{pages}{153--159}.
\newblock


\bibitem[\protect\citeauthoryear{Chang, Yu, Qin, Cheng, and Qiao}{Chang et~al\mbox{.}}{2012}]%
        {chang2012exact}
\bibfield{author}{\bibinfo{person}{Lijun Chang}, \bibinfo{person}{Jeffrey~Xu Yu}, \bibinfo{person}{Lu Qin}, \bibinfo{person}{Hong Cheng}, {and} \bibinfo{person}{Miao Qiao}.} \bibinfo{year}{2012}\natexlab{}.
\newblock \showarticletitle{The Exact Distance to Destination in Undirected World}.
\newblock \bibinfo{journal}{\emph{{VLDB} Journal}} \bibinfo{volume}{21}, \bibinfo{number}{6} (\bibinfo{year}{2012}), \bibinfo{pages}{869--888}.
\newblock


\bibitem[\protect\citeauthoryear{Chen, Liang, and Biros}{Chen et~al\mbox{.}}{2021b}]%
        {rchol}
\bibfield{author}{\bibinfo{person}{Chao Chen}, \bibinfo{person}{Tianyu Liang}, {and} \bibinfo{person}{George Biros}.} \bibinfo{year}{2021}\natexlab{b}.
\newblock \showarticletitle{RCHOL: Randomized Cholesky Factorization for Solving SDD Linear Systems}.
\newblock \bibinfo{journal}{\emph{SIAM Journal on Scientific Computing}} \bibinfo{volume}{43}, \bibinfo{number}{6} (\bibinfo{year}{2021}), \bibinfo{pages}{C411--C438}.
\newblock


\bibitem[\protect\citeauthoryear{Chen and Zhang}{Chen and Zhang}{2007}]%
        {Chen2007ResistanceNormalizedLaplacian}
\bibfield{author}{\bibinfo{person}{Guandong Chen} {and} \bibinfo{person}{Zhongzhi Zhang}.} \bibinfo{year}{2007}\natexlab{}.
\newblock \showarticletitle{Resistance Distance and the Normalized Laplacian Spectrum}.
\newblock \bibinfo{journal}{\emph{Physica A}} \bibinfo{volume}{385}, \bibinfo{number}{2} (\bibinfo{year}{2007}), \bibinfo{pages}{761--772}.
\newblock


\bibitem[\protect\citeauthoryear{Chen, Fu, Jiang, Lo, and Zhang}{Chen et~al\mbox{.}}{2021a}]%
        {chen2021p2h}
\bibfield{author}{\bibinfo{person}{Zitong Chen}, \bibinfo{person}{Ada~Wai{-}Chee Fu}, \bibinfo{person}{Minhao Jiang}, \bibinfo{person}{Eric Lo}, {and} \bibinfo{person}{Pengfei Zhang}.} \bibinfo{year}{2021}\natexlab{a}.
\newblock \showarticletitle{P2H: Efficient Distance Querying on Road Networks by Projected Vertex Separators}. In \bibinfo{booktitle}{\emph{Proceedings of the 2021 {ACM} {SIGMOD} International Conference on Management of Data}}. \bibinfo{pages}{313--325}.
\newblock


\bibitem[\protect\citeauthoryear{Chung}{Chung}{1997}]%
        {Chung1997}
\bibfield{author}{\bibinfo{person}{Fan R.~K. Chung}.} \bibinfo{year}{1997}\natexlab{}.
\newblock \bibinfo{booktitle}{\emph{Spectral Graph Theory}}. \bibinfo{series}{CBMS Regional Conference Series in Mathematics}, Vol.~\bibinfo{volume}{92}.
\newblock \bibinfo{publisher}{American Mathematical Society}.
\newblock
\showISBNx{9780821803158}
\urldef\tempurl%
\url{https://doi.org/10.1090/cbms/092}
\showDOI{\tempurl}


\bibitem[\protect\citeauthoryear{Cohen}{Cohen}{1994}]%
        {cohen1994polylog}
\bibfield{author}{\bibinfo{person}{Edith Cohen}.} \bibinfo{year}{1994}\natexlab{}.
\newblock \showarticletitle{Polylog-time and Near-linear Work Approximation Scheme for Undirected Shortest Paths}. In \bibinfo{booktitle}{\emph{Proceedings of the 26th Annual {ACM} Symposium on Theory of Computing (STOC)}}. \bibinfo{pages}{16--26}.
\newblock


\bibitem[\protect\citeauthoryear{Cohen, Halperin, Kaplan, and Zwick}{Cohen et~al\mbox{.}}{2002}]%
        {cohen2002reachability}
\bibfield{author}{\bibinfo{person}{Edith Cohen}, \bibinfo{person}{Eran Halperin}, \bibinfo{person}{Haim Kaplan}, {and} \bibinfo{person}{Uri Zwick}.} \bibinfo{year}{2002}\natexlab{}.
\newblock \showarticletitle{Reachability and Distance Queries via 2-hop Labels}. In \bibinfo{booktitle}{\emph{Proceedings of the 13th Annual {ACM-SIAM} Symposium on Discrete Algorithms (SODA)}}. \bibinfo{pages}{937--946}.
\newblock


\bibitem[\protect\citeauthoryear{Cohen and Sherman}{Cohen and Sherman}{2014}]%
        {Cohen2014SolvingSDD}
\bibfield{author}{\bibinfo{person}{Michael~B. Cohen} {and} \bibinfo{person}{Jonah Sherman}.} \bibinfo{year}{2014}\natexlab{}.
\newblock \showarticletitle{Solving SDD Linear Systems in Nearly m sqrt{log n} Time}. In \bibinfo{booktitle}{\emph{Proceedings of the 46th ACM Symposium on Theory of Computing (STOC)}}. \bibinfo{pages}{343--352}.
\newblock


\bibitem[\protect\citeauthoryear{Cohen{-}Addad, Dahlgaard, and Wulff{-}Nilsen}{Cohen{-}Addad et~al\mbox{.}}{2017}]%
        {cohenaddad2017fast}
\bibfield{author}{\bibinfo{person}{Vincent Cohen{-}Addad}, \bibinfo{person}{S{\o}ren Dahlgaard}, {and} \bibinfo{person}{Christian Wulff{-}Nilsen}.} \bibinfo{year}{2017}\natexlab{}.
\newblock \showarticletitle{Fast and Compact Exact Distance Oracle for Planar Graphs}. In \bibinfo{booktitle}{\emph{Proceedings of the 58th Annual IEEE Symposium on Foundations of Computer Science (FOCS)}}. \bibinfo{pages}{962--973}.
\newblock


\bibitem[\protect\citeauthoryear{de~Champeaux and Sint}{de~Champeaux and Sint}{1977}]%
        {dechampeaux1977bidirectional}
\bibfield{author}{\bibinfo{person}{Dennis de Champeaux} {and} \bibinfo{person}{Lenie Sint}.} \bibinfo{year}{1977}\natexlab{}.
\newblock \showarticletitle{An Optimality Theorem for a Bi-Directional Heuristic Search Algorithm}. In \bibinfo{booktitle}{\emph{The Computer Journal}}, Vol.~\bibinfo{volume}{20}. \bibinfo{pages}{148--150}.
\newblock


\bibitem[\protect\citeauthoryear{Doyle and Snell}{Doyle and Snell}{1984}]%
        {DoyleSnell1984RandomWalks}
\bibfield{author}{\bibinfo{person}{Peter~G. Doyle} {and} \bibinfo{person}{J.~Laurie Snell}.} \bibinfo{year}{1984}\natexlab{}.
\newblock \bibinfo{booktitle}{\emph{Random Walks and Electric Networks}}.
\newblock \bibinfo{publisher}{Mathematical Association of America}.
\newblock


\bibitem[\protect\citeauthoryear{Elkin and Neiman}{Elkin and Neiman}{2016}]%
        {elkin2016hopsets}
\bibfield{author}{\bibinfo{person}{Michael Elkin} {and} \bibinfo{person}{Ofer Neiman}.} \bibinfo{year}{2016}\natexlab{}.
\newblock \showarticletitle{Hopsets with Constant Hopbound, and Applications to Approximate Shortest Paths}. In \bibinfo{booktitle}{\emph{Proceedings of the 57th Annual IEEE Symposium on Foundations of Computer Science (FOCS)}}. \bibinfo{pages}{128--137}.
\newblock


\bibitem[\protect\citeauthoryear{Farhan, Koehler, and Wang}{Farhan et~al\mbox{.}}{2023}]%
        {farhan2023hierarchical}
\bibfield{author}{\bibinfo{person}{Muhammad Farhan}, \bibinfo{person}{Henning Koehler}, {and} \bibinfo{person}{Qing Wang}.} \bibinfo{year}{2023}\natexlab{}.
\newblock \showarticletitle{Hierarchical Cut Labelling: Scaling Up Distance Queries on Road Networks}.
\newblock \bibinfo{journal}{\emph{Proceedings of the {ACM} on Management of Data}} \bibinfo{volume}{1}, \bibinfo{number}{4} (\bibinfo{year}{2023}), \bibinfo{pages}{244:1--244:25}.
\newblock


\bibitem[\protect\citeauthoryear{Farhan, Koehler, and Wang}{Farhan et~al\mbox{.}}{2025}]%
        {farhan2025dual}
\bibfield{author}{\bibinfo{person}{Muhammad Farhan}, \bibinfo{person}{Henning Koehler}, {and} \bibinfo{person}{Qing Wang}.} \bibinfo{year}{2025}\natexlab{}.
\newblock \showarticletitle{Dual-Hierarchy Labelling: Scaling Up Distance Queries on Dynamic Road Networks}.
\newblock \bibinfo{journal}{\emph{Proceedings of the {ACM} on Management of Data}} \bibinfo{volume}{3}, \bibinfo{number}{1} (\bibinfo{year}{2025}), \bibinfo{pages}{35:1--35:25}.
\newblock


\bibitem[\protect\citeauthoryear{Farzan and Kamali}{Farzan and Kamali}{2011}]%
        {farzan2011compact}
\bibfield{author}{\bibinfo{person}{Arash Farzan} {and} \bibinfo{person}{Shahin Kamali}.} \bibinfo{year}{2011}\natexlab{}.
\newblock \showarticletitle{Compact Navigation and Distance Oracles for Graphs with Small Treewidth}. In \bibinfo{booktitle}{\emph{Proceedings of the 38th International Colloquium on Automata, Languages, and Programming (ICALP)}}. \bibinfo{pages}{268--280}.
\newblock


\bibitem[\protect\citeauthoryear{Fu, Wu, Cheng, and Wong}{Fu et~al\mbox{.}}{2013}]%
        {fu2013islabel}
\bibfield{author}{\bibinfo{person}{Ada~Wai{-}Chee Fu}, \bibinfo{person}{Huanhuan Wu}, \bibinfo{person}{James Cheng}, {and} \bibinfo{person}{Raymond~Chi{-}Wing Wong}.} \bibinfo{year}{2013}\natexlab{}.
\newblock \showarticletitle{{IS-LABEL}: An Independent-Set Based Labeling Scheme for Point-to-Point Distance Querying}.
\newblock \bibinfo{journal}{\emph{{VLDB} Journal}} \bibinfo{volume}{22}, \bibinfo{number}{4} (\bibinfo{year}{2013}), \bibinfo{pages}{457--468}.
\newblock


\bibitem[\protect\citeauthoryear{Gao, Kyng, and Spielman}{Gao et~al\mbox{.}}{2023}]%
        {Gao2023RobustLaplacian}
\bibfield{author}{\bibinfo{person}{Yuan Gao}, \bibinfo{person}{Rasmus Kyng}, {and} \bibinfo{person}{Daniel~A. Spielman}.} \bibinfo{year}{2023}\natexlab{}.
\newblock \showarticletitle{Robust and Practical Solution of Laplacian Equations by Approximate Elimination}.
\newblock \bibinfo{journal}{\emph{arXiv preprint arXiv:2307.05911}} (\bibinfo{year}{2023}).
\newblock


\bibitem[\protect\citeauthoryear{Gavoille, Peleg, Perennes, and Raz}{Gavoille et~al\mbox{.}}{2001}]%
        {gavoille2001distance}
\bibfield{author}{\bibinfo{person}{Cyril Gavoille}, \bibinfo{person}{David Peleg}, \bibinfo{person}{Stephane Perennes}, {and} \bibinfo{person}{Ran Raz}.} \bibinfo{year}{2001}\natexlab{}.
\newblock \showarticletitle{Distance Labeling in Graphs}. In \bibinfo{booktitle}{\emph{Proceedings of the 12th Annual {ACM-SIAM} Symposium on Discrete Algorithms (SODA)}}. \bibinfo{pages}{210--219}.
\newblock


\bibitem[\protect\citeauthoryear{Gleich}{Gleich}{2015}]%
        {Gleich2015SIREV}
\bibfield{author}{\bibinfo{person}{David~F. Gleich}.} \bibinfo{year}{2015}\natexlab{}.
\newblock \showarticletitle{PageRank Beyond the Web}.
\newblock \bibinfo{journal}{\emph{SIAM Rev.}} (\bibinfo{year}{2015}), \bibinfo{pages}{321--363}.
\newblock


\bibitem[\protect\citeauthoryear{Goldberg and Harrelson}{Goldberg and Harrelson}{2005}]%
        {goldberg2005computing}
\bibfield{author}{\bibinfo{person}{Andrew~V. Goldberg} {and} \bibinfo{person}{Chris Harrelson}.} \bibinfo{year}{2005}\natexlab{}.
\newblock \showarticletitle{Computing the Shortest Path: A Search Meets Graph Theory}. In \bibinfo{booktitle}{\emph{Proceedings of the 16th Annual {ACM-SIAM} Symposium on Discrete Algorithms (SODA)}}. \bibinfo{pages}{156--165}.
\newblock


\bibitem[\protect\citeauthoryear{Gupta, Kosowski, and Viennot}{Gupta et~al\mbox{.}}{2019}]%
        {gupta2019exploiting}
\bibfield{author}{\bibinfo{person}{Siddharth Gupta}, \bibinfo{person}{Adrian Kosowski}, {and} \bibinfo{person}{Laurent Viennot}.} \bibinfo{year}{2019}\natexlab{}.
\newblock \showarticletitle{Exploiting Hopsets: Improved Distance Oracles for Graphs of Constant Highway Dimension and Beyond}. In \bibinfo{booktitle}{\emph{Proceedings of the 46th International Colloquium on Automata, Languages, and Programming (ICALP)}}.
\newblock


\bibitem[\protect\citeauthoryear{Haveliwala and Kamvar}{Haveliwala and Kamvar}{2003}]%
        {Haveliwala2003SecondEigenvalueGoogle}
\bibfield{author}{\bibinfo{person}{Taher~H. Haveliwala} {and} \bibinfo{person}{Sepandar~D. Kamvar}.} \bibinfo{year}{2003}\natexlab{}.
\newblock \bibinfo{booktitle}{\emph{The Second Eigenvalue of the Google Matrix}}.
\newblock \bibinfo{type}{{T}echnical {R}eport}. \bibinfo{institution}{Stanford InfoLab}.
\newblock


\bibitem[\protect\citeauthoryear{Horn and Johnson}{Horn and Johnson}{2013}]%
        {HornJohnson2013MatrixAnalysis}
\bibfield{author}{\bibinfo{person}{Roger~A. Horn} {and} \bibinfo{person}{Charles~R. Johnson}.} \bibinfo{year}{2013}\natexlab{}.
\newblock \bibinfo{booktitle}{\emph{Matrix Analysis (2nd ed.)}}.
\newblock \bibinfo{publisher}{Cambridge University Press}.
\newblock


\bibitem[\protect\citeauthoryear{Jiang, Fu, Wong, and Xu}{Jiang et~al\mbox{.}}{2014}]%
        {jiang2014hop}
\bibfield{author}{\bibinfo{person}{Minhao Jiang}, \bibinfo{person}{Ada~Wai{-}Chee Fu}, \bibinfo{person}{Raymond~Chi{-}Wing Wong}, {and} \bibinfo{person}{Yanyan Xu}.} \bibinfo{year}{2014}\natexlab{}.
\newblock \showarticletitle{Hop Doubling Label Indexing for Point-to-Point Distance Querying on Scale-Free Networks}.
\newblock \bibinfo{journal}{\emph{{VLDB} Journal}} \bibinfo{volume}{7}, \bibinfo{number}{12} (\bibinfo{year}{2014}), \bibinfo{pages}{1203--1214}.
\newblock


\bibitem[\protect\citeauthoryear{Jin, Bao, and Zhang}{Jin et~al\mbox{.}}{2019}]%
        {Jin2019ForestDistanceCloseness}
\bibfield{author}{\bibinfo{person}{Yujia Jin}, \bibinfo{person}{Qi Bao}, {and} \bibinfo{person}{Zhongzhi Zhang}.} \bibinfo{year}{2019}\natexlab{}.
\newblock \showarticletitle{Forest Distance Closeness Centrality in Disconnected Graphs}. In \bibinfo{booktitle}{\emph{Proceedings of the IEEE International Conference on Data Mining (ICDM)}}. \bibinfo{pages}{339--348}.
\newblock


\bibitem[\protect\citeauthoryear{Karypis and Kumar}{Karypis and Kumar}{1997}]%
        {METIS}
\bibfield{author}{\bibinfo{person}{George Karypis} {and} \bibinfo{person}{Vipin Kumar}.} \bibinfo{year}{1997}\natexlab{}.
\newblock \showarticletitle{METIS—A Software Package for Partitioning Unstructured Graphs, Partitioning Meshes and Computing Fill-Reducing Ordering of Sparse Matrices}.
\newblock  (\bibinfo{year}{1997}).
\newblock


\bibitem[\protect\citeauthoryear{Karypis and Kumar}{Karypis and Kumar}{1998}]%
        {Mincut1998}
\bibfield{author}{\bibinfo{person}{George Karypis} {and} \bibinfo{person}{Vipin Kumar}.} \bibinfo{year}{1998}\natexlab{}.
\newblock \showarticletitle{A Fast and High Quality Multilevel Scheme for Partitioning Irregular Graphs}.
\newblock \bibinfo{journal}{\emph{SIAM Journal on Scientific Computing}} (\bibinfo{year}{1998}), \bibinfo{pages}{359--392}.
\newblock


\bibitem[\protect\citeauthoryear{Klein and Randi{\'c}}{Klein and Randi{\'c}}{1993}]%
        {Klein1993ResistanceDistance}
\bibfield{author}{\bibinfo{person}{D.~J. Klein} {and} \bibinfo{person}{M. Randi{\'c}}.} \bibinfo{year}{1993}\natexlab{}.
\newblock \showarticletitle{Resistance Distance}.
\newblock \bibinfo{journal}{\emph{Journal of Mathematical Chemistry}} \bibinfo{volume}{12}, \bibinfo{number}{1} (\bibinfo{year}{1993}), \bibinfo{pages}{81--95}.
\newblock


\bibitem[\protect\citeauthoryear{Koehler, Farhan, and Wang}{Koehler et~al\mbox{.}}{2025}]%
        {koehler2025stable}
\bibfield{author}{\bibinfo{person}{Henning Koehler}, \bibinfo{person}{Muhammad Farhan}, {and} \bibinfo{person}{Qing Wang}.} \bibinfo{year}{2025}\natexlab{}.
\newblock \showarticletitle{Stable Tree Labelling for Accelerating Distance Queries on Dynamic Road Networks}. In \bibinfo{booktitle}{\emph{Proceedings of the 2025 International Conference on Extending Database Technology (EDBT)}}. \bibinfo{pages}{477--489}.
\newblock


\bibitem[\protect\citeauthoryear{Kosowski and Viennot}{Kosowski and Viennot}{2017}]%
        {kosowski2017beyond}
\bibfield{author}{\bibinfo{person}{Adrian Kosowski} {and} \bibinfo{person}{Laurent Viennot}.} \bibinfo{year}{2017}\natexlab{}.
\newblock \showarticletitle{Beyond Highway Dimension: Small Distance Labels Using Tree Skeletons}. In \bibinfo{booktitle}{\emph{Proceedings of the 28th Annual {ACM-SIAM} Symposium on Discrete Algorithms (SODA)}}. \bibinfo{pages}{1462--1478}.
\newblock


\bibitem[\protect\citeauthoryear{Kyng and Sachdeva}{Kyng and Sachdeva}{2016}]%
        {KyngSachdeva2016}
\bibfield{author}{\bibinfo{person}{Rasmus Kyng} {and} \bibinfo{person}{Sushant Sachdeva}.} \bibinfo{year}{2016}\natexlab{}.
\newblock \showarticletitle{Approximate Gaussian Elimination for Laplacians: Fast, Sparse, and Simple}. In \bibinfo{booktitle}{\emph{Proceedings of the 57th IEEE Symposium on Foundations of Computer Science (FOCS)}}. \bibinfo{publisher}{IEEE}.
\newblock
\urldef\tempurl%
\url{https://doi.org/10.1109/FOCS.2016.68}
\showDOI{\tempurl}


\bibitem[\protect\citeauthoryear{Lee, Ozdaglar, and Shah}{Lee et~al\mbox{.}}{2013}]%
        {NeurIPS2013Lee}
\bibfield{author}{\bibinfo{person}{Christina~E. Lee}, \bibinfo{person}{Asuman Ozdaglar}, {and} \bibinfo{person}{Devavrat Shah}.} \bibinfo{year}{2013}\natexlab{}.
\newblock \showarticletitle{Computing the stationary distribution locally}. In \bibinfo{booktitle}{\emph{Advances in Neural Information Processing Systems (NeurIPS)}}, Vol.~\bibinfo{volume}{26}. \bibinfo{pages}{1376--1384}.
\newblock


\bibitem[\protect\citeauthoryear{Leskovec and Krevl}{Leskovec and Krevl}{2014}]%
        {Leskovec2014SNAP}
\bibfield{author}{\bibinfo{person}{Jure Leskovec} {and} \bibinfo{person}{Andrej Krevl}.} \bibinfo{year}{2014}\natexlab{}.
\newblock \bibinfo{title}{SNAP Datasets: Stanford Large Network Dataset Collection}.
\newblock
\newblock


\bibitem[\protect\citeauthoryear{Levin, Peres, and Wilmer}{Levin et~al\mbox{.}}{2017}]%
        {Levin2017MarkovChains}
\bibfield{author}{\bibinfo{person}{David~A. Levin}, \bibinfo{person}{Yuval Peres}, {and} \bibinfo{person}{Elizabeth~L. Wilmer}.} \bibinfo{year}{2017}\natexlab{}.
\newblock \bibinfo{booktitle}{\emph{Markov Chains and Mixing Times (2nd ed.)}}.
\newblock \bibinfo{publisher}{American Mathematical Society}.
\newblock


\bibitem[\protect\citeauthoryear{Li and Sachdeva}{Li and Sachdeva}{2023}]%
        {LiSachdeva2023EffectiveResistance}
\bibfield{author}{\bibinfo{person}{Lawrence Li} {and} \bibinfo{person}{Sushant Sachdeva}.} \bibinfo{year}{2023}\natexlab{}.
\newblock \showarticletitle{A New Approach to Estimating Effective Resistances and Counting Spanning Trees in Expander Graphs}. In \bibinfo{booktitle}{\emph{Proceedings of the 2023 ACM-SIAM Symposium on Discrete Algorithms (SODA)}}. \bibinfo{pages}{2728--2745}.
\newblock


\bibitem[\protect\citeauthoryear{Li, Qiao, Qin, Zhang, Chang, and Lin}{Li et~al\mbox{.}}{2020}]%
        {li2020scaling}
\bibfield{author}{\bibinfo{person}{Wentao Li}, \bibinfo{person}{Miao Qiao}, \bibinfo{person}{Lu Qin}, \bibinfo{person}{Ying Zhang}, \bibinfo{person}{Lijun Chang}, {and} \bibinfo{person}{Xuemin Lin}.} \bibinfo{year}{2020}\natexlab{}.
\newblock \showarticletitle{Scaling Up Distance Labeling on Graphs with Core-Periphery Properties}.
\newblock \bibinfo{journal}{\emph{Proceedings of the {ACM} on Management of Data}} \bibinfo{volume}{1}, \bibinfo{number}{1} (\bibinfo{year}{2020}), \bibinfo{pages}{1367--1381}.
\newblock


\bibitem[\protect\citeauthoryear{Liao, Li, Dai, Chen, Qin, and Wang}{Liao et~al\mbox{.}}{2023}]%
        {Landmark2023Liao}
\bibfield{author}{\bibinfo{person}{Meihao Liao}, \bibinfo{person}{Rong-Hua Li}, \bibinfo{person}{Qiangqiang Dai}, \bibinfo{person}{Hongyang Chen}, \bibinfo{person}{Hongchao Qin}, {and} \bibinfo{person}{Guoren Wang}.} \bibinfo{year}{2023}\natexlab{}.
\newblock \showarticletitle{Efficient Resistance Distance Computation: The Power of Landmark-based Approaches}.
\newblock \bibinfo{journal}{\emph{SIGMOD}} (\bibinfo{year}{2023}).
\newblock


\bibitem[\protect\citeauthoryear{Liao, Zhou, Li, Dai, Chen, and Wang}{Liao et~al\mbox{.}}{2024}]%
        {Liao-ER-Index}
\bibfield{author}{\bibinfo{person}{Meihao Liao}, \bibinfo{person}{Junjie Zhou}, \bibinfo{person}{Rong-Hua Li}, \bibinfo{person}{Qiangqiang Dai}, \bibinfo{person}{Hongyang Chen}, {and} \bibinfo{person}{Guoren Wang}.} \bibinfo{year}{2024}\natexlab{}.
\newblock \showarticletitle{Efficient and Provable Effective Resistance Computation on Large Graphs: An Index-based Approach}.
\newblock \bibinfo{journal}{\emph{Proc. ACM Manag. Data}} (\bibinfo{year}{2024}).
\newblock


\bibitem[\protect\citeauthoryear{Lipman, Rustamov, and Funkhouser}{Lipman et~al\mbox{.}}{2010}]%
        {TOG2010Lipman}
\bibfield{author}{\bibinfo{person}{Yaron Lipman}, \bibinfo{person}{Raif~M. Rustamov}, {and} \bibinfo{person}{Thomas~A. Funkhouser}.} \bibinfo{year}{2010}\natexlab{}.
\newblock \showarticletitle{Biharmonic distance}.
\newblock \bibinfo{journal}{\emph{ACM Transactions on Graphics}} \bibinfo{volume}{29}, \bibinfo{number}{1} (\bibinfo{year}{2010}), \bibinfo{pages}{1--11}.
\newblock


\bibitem[\protect\citeauthoryear{Liu, Zehmakan, and Zhang}{Liu et~al\mbox{.}}{2024}]%
        {LiuKDD2024}
\bibfield{author}{\bibinfo{person}{Changan Liu}, \bibinfo{person}{Ahad~N. Zehmakan}, {and} \bibinfo{person}{Zhongzhi Zhang}.} \bibinfo{year}{2024}\natexlab{}.
\newblock \showarticletitle{Fast Query of Biharmonic Distance in Networks}. In \bibinfo{booktitle}{\emph{Proceedings of the 30th ACM SIGKDD Conference on Knowledge Discovery and Data Mining}} \emph{(\bibinfo{series}{KDD '24})}. \bibinfo{pages}{1887–1897}.
\newblock


\bibitem[\protect\citeauthoryear{Liu and Yu}{Liu and Yu}{2023}]%
        {Liao2024}
\bibfield{author}{\bibinfo{person}{Zhiqiang Liu} {and} \bibinfo{person}{Wenjian Yu}.} \bibinfo{year}{2023}\natexlab{}.
\newblock \showarticletitle{Computing Effective Resistances on Large Graphs Based on Approximate Inverse of Cholesky Factor}. In \bibinfo{booktitle}{\emph{2023 Design, Automation \& Test in Europe Conference \& Exhibition (DATE)}}. \bibinfo{pages}{1--6}.
\newblock


\bibitem[\protect\citeauthoryear{L{\"u} and Zhou}{L{\"u} and Zhou}{2011}]%
        {Lu2011LinkPredictionSurvey}
\bibfield{author}{\bibinfo{person}{Linyuan L{\"u}} {and} \bibinfo{person}{Tao Zhou}.} \bibinfo{year}{2011}\natexlab{}.
\newblock \showarticletitle{Link Prediction in Complex Networks: A Survey}.
\newblock \bibinfo{journal}{\emph{Physica A}} \bibinfo{volume}{390}, \bibinfo{number}{6} (\bibinfo{year}{2011}), \bibinfo{pages}{1150--1170}.
\newblock


\bibitem[\protect\citeauthoryear{Mahoney}{Mahoney}{2012}]%
        {mahoney2012local}
\bibfield{author}{\bibinfo{person}{Michael~W. Mahoney}.} \bibinfo{year}{2012}\natexlab{}.
\newblock \showarticletitle{A Local Spectral Method for Graphs: With Applications to Semi-Supervised Learning and Partitioning}.
\newblock \bibinfo{journal}{\emph{Journal of Machine Learning Research (JMLR)}}  \bibinfo{volume}{13} (\bibinfo{year}{2012}), \bibinfo{pages}{2339--2364}.
\newblock


\bibitem[\protect\citeauthoryear{Motwani and Raghavan}{Motwani and Raghavan}{1995}]%
        {MotwaniRaghavan1995RandomizedAlgorithms}
\bibfield{author}{\bibinfo{person}{Rajeev Motwani} {and} \bibinfo{person}{Prabhakar Raghavan}.} \bibinfo{year}{1995}\natexlab{}.
\newblock \showarticletitle{Randomized Algorithms}.
\newblock \bibinfo{journal}{\emph{SIGACT News}} \bibinfo{volume}{26}, \bibinfo{number}{3} (\bibinfo{year}{1995}), \bibinfo{pages}{48--50}.
\newblock


\bibitem[\protect\citeauthoryear{Ouyang, Qin, Chang, Lin, Zhang, and Zhu}{Ouyang et~al\mbox{.}}{2018}]%
        {ouyang2018when}
\bibfield{author}{\bibinfo{person}{Dian Ouyang}, \bibinfo{person}{Lu Qin}, \bibinfo{person}{Lijun Chang}, \bibinfo{person}{Xuemin Lin}, \bibinfo{person}{Ying Zhang}, {and} \bibinfo{person}{Qing Zhu}.} \bibinfo{year}{2018}\natexlab{}.
\newblock \showarticletitle{When Hierarchy Meets 2-Hop-Labeling: Efficient Shortest Distance Queries on Road Networks}. In \bibinfo{booktitle}{\emph{Proceedings of the 2018 {ACM} {SIGMOD} International Conference on Management of Data}}. \bibinfo{pages}{709--724}.
\newblock


\bibitem[\protect\citeauthoryear{Ouyang, Yuan, Qin, Chang, Zhang, and Lin}{Ouyang et~al\mbox{.}}{2020}]%
        {ouyang2020efficient}
\bibfield{author}{\bibinfo{person}{Dian Ouyang}, \bibinfo{person}{Long Yuan}, \bibinfo{person}{Lu Qin}, \bibinfo{person}{Lijun Chang}, \bibinfo{person}{Ying Zhang}, {and} \bibinfo{person}{Xuemin Lin}.} \bibinfo{year}{2020}\natexlab{}.
\newblock \showarticletitle{Efficient Shortest Path Index Maintenance on Dynamic Road Networks with Theoretical Guarantees}.
\newblock \bibinfo{journal}{\emph{Proceedings of the {VLDB} Endowment}} \bibinfo{volume}{13}, \bibinfo{number}{5} (\bibinfo{year}{2020}), \bibinfo{pages}{602--615}.
\newblock


\bibitem[\protect\citeauthoryear{Peng, Lopatta, Yoshida, and Goranci}{Peng et~al\mbox{.}}{2021}]%
        {KDD2021Peng}
\bibfield{author}{\bibinfo{person}{Pan Peng}, \bibinfo{person}{Daniel Lopatta}, \bibinfo{person}{Yuichi Yoshida}, {and} \bibinfo{person}{Gramoz Goranci}.} \bibinfo{year}{2021}\natexlab{}.
\newblock \showarticletitle{Local algorithms for estimating effective resistance}. In \bibinfo{booktitle}{\emph{Proceedings of the 27th ACM SIGKDD Conference on Knowledge Discovery and Data Mining (KDD)}}. \bibinfo{publisher}{ACM}, \bibinfo{pages}{1329--1338}.
\newblock


\bibitem[\protect\citeauthoryear{Predari et~al\mbox{.}}{Predari et~al\mbox{.}}{2023}]%
        {predari2023greedy}
\bibfield{author}{\bibinfo{person}{M. Predari} {et~al\mbox{.}}} \bibinfo{year}{2023}\natexlab{}.
\newblock \showarticletitle{Greedy optimization of resistance-based graph robustness}.
\newblock \bibinfo{journal}{\emph{Social Network Analysis and Mining}} \bibinfo{volume}{13}, \bibinfo{number}{1} (\bibinfo{year}{2023}).
\newblock


\bibitem[\protect\citeauthoryear{Robertson and Seymour}{Robertson and Seymour}{1991}]%
        {TreeDecomposion1991}
\bibfield{author}{\bibinfo{person}{Neil Robertson} {and} \bibinfo{person}{P.D Seymour}.} \bibinfo{year}{1991}\natexlab{}.
\newblock \showarticletitle{Graph minors. X. Obstructions to tree-decomposition}.
\newblock \bibinfo{journal}{\emph{Journal of Combinatorial Theory, Series B}} (\bibinfo{year}{1991}), \bibinfo{pages}{153--190}.
\newblock


\bibitem[\protect\citeauthoryear{Rossi and Ahmed}{Rossi and Ahmed}{2015}]%
        {dataset-DIMACS}
\bibfield{author}{\bibinfo{person}{Ryan~A. Rossi} {and} \bibinfo{person}{Nesreen~K. Ahmed}.} \bibinfo{year}{2015}\natexlab{}.
\newblock \showarticletitle{The Network Data Repository with Interactive Graph Analytics and Visualization}. In \bibinfo{booktitle}{\emph{AAAI}}.
\newblock


\bibitem[\protect\citeauthoryear{Rozemberczki, Allen, and Sarkar}{Rozemberczki et~al\mbox{.}}{2021}]%
        {chameleon2021}
\bibfield{author}{\bibinfo{person}{Benedek Rozemberczki}, \bibinfo{person}{Carl Allen}, {and} \bibinfo{person}{Rik Sarkar}.} \bibinfo{year}{2021}\natexlab{}.
\newblock \showarticletitle{Multi-scale Attributed Node Embedding}.
\newblock \bibinfo{journal}{\emph{Journal of Complex Networks}} (\bibinfo{year}{2021}), \bibinfo{pages}{cnab014}.
\newblock


\bibitem[\protect\citeauthoryear{Shimada, Hirata, Ikeguchi, and Aihara}{Shimada et~al\mbox{.}}{2016}]%
        {Shimada2016GraphDistanceComplex}
\bibfield{author}{\bibinfo{person}{Yutaka Shimada}, \bibinfo{person}{Yoshito Hirata}, \bibinfo{person}{Tohru Ikeguchi}, {and} \bibinfo{person}{Kazuyuki Aihara}.} \bibinfo{year}{2016}\natexlab{}.
\newblock \showarticletitle{Graph Distance for Complex Networks}.
\newblock \bibinfo{journal}{\emph{Scientific Reports}}  \bibinfo{volume}{6} (\bibinfo{year}{2016}), \bibinfo{pages}{34944}.
\newblock


\bibitem[\protect\citeauthoryear{Shur, Huang, and Gleich}{Shur et~al\mbox{.}}{2023}]%
        {shur2023pagerank_embedding}
\bibfield{author}{\bibinfo{person}{Disha Shur}, \bibinfo{person}{Yufan Huang}, {and} \bibinfo{person}{David~F. Gleich}.} \bibinfo{year}{2023}\natexlab{}.
\newblock \showarticletitle{A flexible PageRank-based graph embedding framework}.
\newblock \bibinfo{journal}{\emph{Applied and Computational Topology}} \bibinfo{volume}{7}, \bibinfo{number}{1} (\bibinfo{year}{2023}).
\newblock


\bibitem[\protect\citeauthoryear{Spielman and Srivastava}{Spielman and Srivastava}{2008}]%
        {spielman2008sparsification}
\bibfield{author}{\bibinfo{person}{Daniel~A. Spielman} {and} \bibinfo{person}{Nikhil Srivastava}.} \bibinfo{year}{2008}\natexlab{}.
\newblock \showarticletitle{Graph Sparsification by Effective Resistances}. In \bibinfo{booktitle}{\emph{Proceedings of the 40th Annual ACM Symposium on Theory of Computing (STOC)}}. \bibinfo{pages}{563--568}.
\newblock


\bibitem[\protect\citeauthoryear{Spielman and Teng}{Spielman and Teng}{2014}]%
        {SpielmanTeng2014NearlyLinearSDD}
\bibfield{author}{\bibinfo{person}{Daniel~A. Spielman} {and} \bibinfo{person}{Shang-Hua Teng}.} \bibinfo{year}{2014}\natexlab{}.
\newblock \showarticletitle{Nearly Linear Time Algorithms for Preconditioning and Solving Symmetric, Diagonally Dominant Linear Systems}.
\newblock \bibinfo{journal}{\emph{SIAM J. Matrix Anal. Appl.}} \bibinfo{volume}{35}, \bibinfo{number}{3} (\bibinfo{year}{2014}), \bibinfo{pages}{835--885}.
\newblock


\bibitem[\protect\citeauthoryear{Tong, Faloutsos, and Pan}{Tong et~al\mbox{.}}{2006}]%
        {ICDM2006Tong}
\bibfield{author}{\bibinfo{person}{Hanghang Tong}, \bibinfo{person}{Christos Faloutsos}, {and} \bibinfo{person}{Jia-Yu Pan}.} \bibinfo{year}{2006}\natexlab{}.
\newblock \showarticletitle{Fast random walk with restart and its applications}. In \bibinfo{booktitle}{\emph{Proceedings of the Sixth International Conference on Data Mining (ICDM)}}. \bibinfo{publisher}{IEEE}, \bibinfo{pages}{613--622}.
\newblock


\bibitem[\protect\citeauthoryear{Topping, Di~Giovanni, Chamberlain, Dong, and Bronstein}{Topping et~al\mbox{.}}{2022}]%
        {ICLR2022Jake}
\bibfield{author}{\bibinfo{person}{Jake Topping}, \bibinfo{person}{Francesco Di~Giovanni}, \bibinfo{person}{Benjamin Chamberlain}, \bibinfo{person}{Xiaowen Dong}, {and} \bibinfo{person}{Michael Bronstein}.} \bibinfo{year}{2022}\natexlab{}.
\newblock \bibinfo{title}{Understanding over-squashing and bottlenecks on graphs via curvature}.
\newblock
\newblock


\bibitem[\protect\citeauthoryear{Tran, Liu, Lee, and Kong}{Tran et~al\mbox{.}}{2019}]%
        {Tran2019SignedDistanceRecommender}
\bibfield{author}{\bibinfo{person}{Thanh Tran}, \bibinfo{person}{Xinyue Liu}, \bibinfo{person}{Kyumin Lee}, {and} \bibinfo{person}{Xiangnan Kong}.} \bibinfo{year}{2019}\natexlab{}.
\newblock \showarticletitle{Signed Distance-based Deep Memory Recommender}. In \bibinfo{booktitle}{\emph{Proceedings of The Web Conference (WWW)}}. \bibinfo{pages}{1841--1852}.
\newblock


\bibitem[\protect\citeauthoryear{Tsitsulin, Munkhoeva, and Perozzi}{Tsitsulin et~al\mbox{.}}{2020}]%
        {Tsitsulin2020SLaQ}
\bibfield{author}{\bibinfo{person}{Anton Tsitsulin}, \bibinfo{person}{Marina Munkhoeva}, {and} \bibinfo{person}{Bryan Perozzi}.} \bibinfo{year}{2020}\natexlab{}.
\newblock \showarticletitle{Just SLaQ When You Approximate: Accurate Spectral Distances for Web-Scale Graphs}. In \bibinfo{booktitle}{\emph{Proceedings of The Web Conference (WWW)}}. \bibinfo{pages}{2697--2703}.
\newblock


\bibitem[\protect\citeauthoryear{Tyloo, Coletta, and Jacquod}{Tyloo et~al\mbox{.}}{2018}]%
        {Tyloo2018RobustnessSynchrony}
\bibfield{author}{\bibinfo{person}{Melvyn Tyloo}, \bibinfo{person}{Tommaso Coletta}, {and} \bibinfo{person}{Philippe Jacquod}.} \bibinfo{year}{2018}\natexlab{}.
\newblock \showarticletitle{Robustness of Synchrony in Complex Networks and Generalized Kirchhoff Indices}.
\newblock \bibinfo{journal}{\emph{Physical Review Letters}} \bibinfo{volume}{120}, \bibinfo{number}{8} (\bibinfo{year}{2018}), \bibinfo{pages}{084101}.
\newblock


\bibitem[\protect\citeauthoryear{Tyloo, Pagnier, and Jacquod}{Tyloo et~al\mbox{.}}{2019}]%
        {Tyloo2019KeyPlayerResistance}
\bibfield{author}{\bibinfo{person}{M. Tyloo}, \bibinfo{person}{L. Pagnier}, {and} \bibinfo{person}{P. Jacquod}.} \bibinfo{year}{2019}\natexlab{}.
\newblock \showarticletitle{The Key Player Problem in Complex Oscillator Networks and Electric Power Grids: Resistance Centralities Identify Local Vulnerabilities}.
\newblock \bibinfo{journal}{\emph{Science Advances}} \bibinfo{volume}{5}, \bibinfo{number}{11} (\bibinfo{year}{2019}), \bibinfo{pages}{eaaw8359}.
\newblock


\bibitem[\protect\citeauthoryear{Wang, Yang, Xiao, Wei, and Yang}{Wang et~al\mbox{.}}{2017}]%
        {KDD2017Wang}
\bibfield{author}{\bibinfo{person}{Sibo Wang}, \bibinfo{person}{Renchi Yang}, \bibinfo{person}{Xiaokui Xiao}, \bibinfo{person}{Zhewei Wei}, {and} \bibinfo{person}{Yin Yang}.} \bibinfo{year}{2017}\natexlab{}.
\newblock \showarticletitle{FORA: simple and effective approximate single-source personalized PageRank}. In \bibinfo{booktitle}{\emph{Proceedings of the 23rd ACM SIGKDD International Conference on Knowledge Discovery and Data Mining (KDD)}}. \bibinfo{publisher}{ACM}, \bibinfo{pages}{505--514}.
\newblock


\bibitem[\protect\citeauthoryear{Wei}{Wei}{2010}]%
        {wei2010tedi}
\bibfield{author}{\bibinfo{person}{Fang Wei}.} \bibinfo{year}{2010}\natexlab{}.
\newblock \showarticletitle{{TEDI}: Efficient Shortest Path Query Answering on Graphs}. In \bibinfo{booktitle}{\emph{Proceedings of the 2010 {ACM} {SIGMOD} International Conference on Management of Data}}. \bibinfo{pages}{99--110}.
\newblock


\bibitem[\protect\citeauthoryear{Xie and Lang}{Xie and Lang}{2015}]%
        {xie2015edgeppr}
\bibfield{author}{\bibinfo{person}{W. Xie} {and} \bibinfo{person}{K. Lang}.} \bibinfo{year}{2015}\natexlab{}.
\newblock \showarticletitle{Edge-Weighted Personalized PageRank: Breaking A Decade of Barriers}. In \bibinfo{booktitle}{\emph{Proceedings of the 24th International Conference on World Wide Web (WWW)}}. \bibinfo{pages}{727--737}.
\newblock


\bibitem[\protect\citeauthoryear{Yang}{Yang}{2022}]%
        {Yang2022BipartiteSimilaritySearch}
\bibfield{author}{\bibinfo{person}{Renchi Yang}.} \bibinfo{year}{2022}\natexlab{}.
\newblock \showarticletitle{Efficient and Effective Similarity Search over Bipartite Graphs}. In \bibinfo{booktitle}{\emph{Proceedings of The Web Conference (WWW)}}. \bibinfo{pages}{308--318}.
\newblock


\bibitem[\protect\citeauthoryear{Yang and Klein}{Yang and Klein}{2013}]%
        {YangKlein2013RecursionResistance}
\bibfield{author}{\bibinfo{person}{Yujun Yang} {and} \bibinfo{person}{Douglas~J. Klein}.} \bibinfo{year}{2013}\natexlab{}.
\newblock \showarticletitle{A Recursion Formula for Resistance Distances and Its Applications}.
\newblock \bibinfo{journal}{\emph{Discrete Applied Mathematics}} \bibinfo{volume}{161}, \bibinfo{number}{16-17} (\bibinfo{year}{2013}), \bibinfo{pages}{2702--2715}.
\newblock


\bibitem[\protect\citeauthoryear{Yi, Shan, Li, and Zhang}{Yi et~al\mbox{.}}{2018a}]%
        {IJCAI2018Yi}
\bibfield{author}{\bibinfo{person}{Yuhao Yi}, \bibinfo{person}{Liren Shan}, \bibinfo{person}{Huan Li}, {and} \bibinfo{person}{Zhongzhi Zhang}.} \bibinfo{year}{2018}\natexlab{a}.
\newblock \showarticletitle{Biharmonic distance related centrality for edges in weighted networks}. In \bibinfo{booktitle}{\emph{Proceedings of the 27th International Joint Conference on Artificial Intelligence (IJCAI-18)}}. \bibinfo{pages}{3620--3626}.
\newblock


\bibitem[\protect\citeauthoryear{Yi, Yang, Zhang, and Patterson}{Yi et~al\mbox{.}}{2018b}]%
        {ACC2018Yi}
\bibfield{author}{\bibinfo{person}{Yuhao Yi}, \bibinfo{person}{Bingjia Yang}, \bibinfo{person}{Zhongzhi Zhang}, {and} \bibinfo{person}{Stacy Patterson}.} \bibinfo{year}{2018}\natexlab{b}.
\newblock \showarticletitle{Biharmonic distance and performance of second-order consensus networks with stochastic disturbances}. In \bibinfo{booktitle}{\emph{Proceedings of the American Control Conference (ACC)}}. \bibinfo{pages}{4943--4950}.
\newblock


\bibitem[\protect\citeauthoryear{Yi, Yang, Zhang, Zhang, and Patterson}{Yi et~al\mbox{.}}{2022}]%
        {TIT2022Yi}
\bibfield{author}{\bibinfo{person}{Yuhao Yi}, \bibinfo{person}{Bingjia Yang}, \bibinfo{person}{Zuobai Zhang}, \bibinfo{person}{Zhongzhi Zhang}, {and} \bibinfo{person}{Stacy Patterson}.} \bibinfo{year}{2022}\natexlab{}.
\newblock \showarticletitle{Biharmonic distance-based performance metric for second-order noisy consensus networks}.
\newblock \bibinfo{journal}{\emph{IEEE Transactions on Information Theory}}  \bibinfo{volume}{68} (\bibinfo{year}{2022}), \bibinfo{pages}{1220--1236}.
\newblock


\bibitem[\protect\citeauthoryear{Zhang and Yu}{Zhang and Yu}{2022}]%
        {zhang2022relative}
\bibfield{author}{\bibinfo{person}{Yikai Zhang} {and} \bibinfo{person}{Jeffrey~Xu Yu}.} \bibinfo{year}{2022}\natexlab{}.
\newblock \showarticletitle{Relative Subboundedness of Contraction Hierarchy and Hierarchical 2-Hop Index in Dynamic Road Networks}. In \bibinfo{booktitle}{\emph{Proceedings of the 2022 {ACM} {SIGMOD} International Conference on Management of Data}}. \bibinfo{pages}{1992--2005}.
\newblock


\bibitem[\protect\citeauthoryear{Zheng, Ma, Wan, Gao, Huang, Zhou, and Jensen}{Zheng et~al\mbox{.}}{2023}]%
        {zheng2023reinforcement}
\bibfield{author}{\bibinfo{person}{Bolong Zheng}, \bibinfo{person}{Yong Ma}, \bibinfo{person}{Jingyi Wan}, \bibinfo{person}{Yongyong Gao}, \bibinfo{person}{Kai Huang}, \bibinfo{person}{Xiaofang Zhou}, {and} \bibinfo{person}{Christian~S. Jensen}.} \bibinfo{year}{2023}\natexlab{}.
\newblock \showarticletitle{Reinforcement Learning Based Tree Decomposition for Distance Querying in Road Networks}. In \bibinfo{booktitle}{\emph{Proceedings of the 2023 IEEE 39th International Conference on Data Engineering (ICDE)}}. \bibinfo{pages}{1678--1690}.
\newblock


\end{thebibliography}

\end{document}